\documentclass[11pt,letterpaper,reqno]{amsart}
\title[Geometry and Dynamics of Gaussian Wave Packets and their Wigner Transforms]{Geometry and Dynamics of Gaussian Wave Packets\\and their Wigner Transforms}
\author{Tomoki Ohsawa}
\address{Department of Mathematical Sciences, The University of Texas at Dallas, 800 W Campbell Rd, Richardson, TX 75080-3021, United States}
\email{tomoki@utdallas.edu}
\author{Cesare Tronci}
\address{Department of Mathematics, University of Surrey, Guildford, GU2 7XH, United Kingdom}
\email{c.tronci@surrey.ac.uk}

\date{\today}


\usepackage{graphicx,wrapfig,amssymb,paralist,subfigure,caption}
\usepackage{amsmath,framed,mathrsfs}
\usepackage[margin=1in, marginpar=.5in]{geometry}

\usepackage[numbers,sort&compress]{natbib}
\usepackage[colorlinks]{hyperref}

\usepackage{tikz-cd}

\theoremstyle{plain}
\newtheorem{theorem}{Theorem}[section]
\newtheorem{corollary}[theorem]{Corollary}

\newtheorem{proposition}[theorem]{Proposition}

\theoremstyle{definition}

\theoremstyle{remark}
\newtheorem{remark}[theorem]{Remark}

\def\od#1#2{\dfrac{\mathrm{d}#1}{\mathrm{d}#2}}
\def\pd#1#2{\dfrac{\partial #1}{\partial #2}}
\def\Fd#1#2{\dfrac{\delta #1}{\delta #2}}
\def\tpd#1#2{\partial #1/\partial #2}

\def\parentheses#1{{\left(#1\right)}}
\def\brackets#1{{\left[#1\right]}}
\def\braces#1{{\left\{#1\right\}}}

\def\tr{\mathop{\mathrm{tr}}\nolimits}

\def\pr{\mathop{\mathrm{pr}}\nolimits}

\def\norm#1{{\left\|#1\right\|}}

\def\DS{\displaystyle}
\def\R{\mathbb{R}}
\def\C{\mathbb{C}}

\def\defeq{\mathrel{\mathop:}=}
\def\eqdef{=\mathrel{\mathop:}}
\def\setdef#1#2{{\left\{ #1 \ |\ #2 \right\}}}
\def\ip#1#2{{\left\langle#1,#2\right\rangle}}
\def\tip#1#2{{\langle#1,#2\rangle}}
\def\exval#1{{\left\langle#1\right\rangle}}

\newcommand{\id}{\operatorname{id}}
\renewcommand{\Re}{\operatorname{Re}}
\renewcommand{\Im}{\operatorname{Im}}
\def\eps{\hbar}
\def\Mat{\mathsf{M}}

\def\Sp{\mathsf{Sp}}
\def\Orth{\mathsf{O}}
\def\U{\mathsf{U}}

\def\sp{\mathfrak{sp}}

\def\sym{\mathsf{sym}}

\newenvironment{tbmatrix}{\left[\begin{smallmatrix}}{\end{smallmatrix}\right]}

\def\d{\mathbf{d}}
\def\ins#1{\mathbf{i}_{#1}}

\def\PB#1#2{\left\{#1,#2\right\}}
\newcommand\Ad{\operatorname{Ad}}
\newcommand\ad{\operatorname{ad}}

\def\rmi{{\rm i}}

\begin{document}

\footskip=.6in

\begin{abstract}
  We find a relationship between the dynamics of the Gaussian wave packet and the dynamics of the corresponding Gaussian Wigner function from the Hamiltonian/symplectic point of view.
  The main result states that the momentum map corresponding to the natural action of the symplectic group on the Siegel upper half space yields the covariance matrix of the corresponding Gaussian Wigner function.
  This fact, combined with Kostant's coadjoint orbit covering theorem, establishes a symplectic/Poisson-geometric connection between the two dynamics.
  The Hamiltonian formulation naturally gives rise to corrections to the potential terms in the dynamics of both the wave packet and the Wigner function, thereby resulting in slightly different sets of equations from the conventional classical ones.
  We numerically investigate the effect of the correction term and demonstrate that it improves the accuracy of the dynamics as an approximation to the dynamics of expectation values of observables.
\end{abstract}

\maketitle

\section{Introduction}
\subsection{Background}
Coherent states play a crucial role in quantum dynamics, and their mathematical properties have been exploited over the decades in many different fields, especially quantum optics and chemical physics; see, e.g., \citet{Be2007,BiMo1991,BoTr2016,CoRo2012}.
This is due to the fact that coherent states behave like classical states, in the sense that the expectation values of the quantum canonical operators undergo classical Hamiltonian dynamics; see, e.g., \citet{Go2011,CoRo2012}.
Indeed, it is well known that, for quadratic Hamiltonians $h$ defined on $T^{*}\R^{d} = \R^{2d}$, the time evolution equation of the Wigner function becomes identical to the Liouville equation
\begin{equation}
  \label{eq:Liouville}
  \pd{f}{t} = -\PB{f}{h}_{\R^{2d}}
\end{equation}
for the corresponding classical system, where $\PB{\,\cdot\,}{\,\cdot\,}_{\R^{2d}}$ is the canonical Poisson bracket on $T^{*}\R^{d} = \R^{2d}$, i.e., for any $f, g \in C^{\infty}(\R^{2d})$,
\begin{equation*}
  \PB{f}{g}_{\R^{2d}} \defeq \pd{f}{q^{i}}\pd{g}{p_{i}} - \pd{g}{q^{i}}\pd{f}{p_{i}}
\end{equation*}
using Einstein's summation convention.
Besides their interesting properties relating classical and quantum systems, coherent states have always attracted much attention due to their intriguing geometric properties.
Specifically, coherent states are defined (up to phase factors) as orbits of the representation of the Heisenberg group on the $L^2$ space of wave functions~\cite{Go2011}. In particular, it is customary to select the particular orbit corresponding to the Gaussian wave function arising as the vacuum (or ground) state solution of the harmonic oscillator.
This interpretation of coherent states in terms of group orbits led \citet{Pe1986} to define generalized coherent states in terms of orbits corresponding to other group representations. For example, spin coherent states are orbits of $\mathsf{SU}(2)$ for its natural representation on the space of Pauli spinors.
Also, squeezed coherent states or Gaussian wave packets are orbits of the Lie group---sometimes called the Schr\"odinger group---given by the semidirect product of the metaplectic group and the Heisenberg group \cite{Li1986, Go2011}:
Applying the representation of the Schr\"odinger group on the vacuum state of the harmonic oscillator yields the squeezed coherent state or the Gaussian wave packet, which is among the most studied quantum states in the literature; see e.g., \citet{He1975a,He1976b}, \citet{Li1986}, \citet{Ha1980, Ha1981, Hagedorn1985, Ha1998}, \citet{CoRo2012}.

The emergence of the metaplectic group in the structure of the Gaussian wave packet makes their mathematical study particularly interesting and also somewhat intricate, due to the form of the metaplectic representation~\cite{Go2011,CoRo2012}.
However, in the phase space picture of quantum mechanics, the subtlety of the metaplectic representation disappears and one may work with the corresponding symplectic matrices instead: Indeed, the symplectic group possesses a natural action on functions on the phase space.
The Wigner transform of a Gaussian wave packet is a Gaussian function in the phase space that is entirely characterized by its mean (phase space center) $z$ and symplectic covariance matrix $\Sigma$; see \eqref{eq:Wigner-psi_0} below.
It is common in the literature to describe the dynamics of the mean $z$ by the classical Hamiltonian system and that of the covariance matrix by the congruence transformation $\Sigma \mapsto S\Sigma S^T$ given by the symplectic matrix $S$, which in turn evolves according to the linearization of the classical Hamiltonian system. Upon extending to a more general positive-definite covariance matrix $\Sigma$, this also applies to {\it any} Gaussian Wigner function on phase space~\cite{BoTr2016}.

\subsection{Motivation}
The main focus of this paper is the geometry and dynamics of the Gaussian wave packet
\begin{equation}
  \label{eq:chi}
  \chi_{0}(x) \defeq \exp\braces{ \frac{{\rm i}}{\hbar}\brackets{ \frac{1}{2}(x - q)^{T}(\mathcal{A} + {\rm i}\mathcal{B})(x - q) + p \cdot (x - q) + (\phi + {\rm i}\delta) } }
\end{equation}
and its Wigner transform.
We are particularly interested in establishing a connection between the dynamics of the two in a symplectic/Poisson-geometric manner.

The above Gaussian wave packet~\eqref{eq:chi} is parametrized by $(q,p) \in T^{*}\R^{d} \cong \R^{2d}$, $\phi \in \mathbb{S}^{1}$, $\delta \in \R$, and $\mathcal{C} \defeq \mathcal{A} + {\rm i}\mathcal{B} \in \mathbb{H}_{d}$, where $\mathbb{H}_{d}$ is the set of symmetric $d \times d$ complex matrices (symmetric in the real sense) with positive-definite imaginary parts, i.e.,
\begin{equation}
  \label{eq:H_d}
  \mathbb{H}_{d} \defeq 
  \setdef{ \mathcal{C} = \mathcal{A} + {\rm i}\mathcal{B} \in \Mat(d,\mathbb{C}) }{ \mathcal{A}, \mathcal{B} \in \sym(d,\R),\, \mathcal{B} > 0 },
\end{equation}
and is called the {\em Siegel upper half space}~\cite{Si1943}; $\Mat(d,\mathbb{C})$ and $\sym(d,\R)$ stand for the set of $d \times d$ complex matrices and the set of $d \times d$ symmetric real matrices, respectively.
A practical significance of the Gaussian wave packet~\eqref{eq:chi} is that it is an {\em exact} solution of the time-dependent Schr\"odinger equation with quadratic Hamiltonians if the parameters $(q, p, \mathcal{A}, \mathcal{B}, \phi, \delta)$, as functions of the time, satisfy a certain set of ODEs.
It also possesses other nice properties as approximations to the exact solution; see \citet{He1975a,He1976b} and \citet{Ha1980, Ha1981, Hagedorn1985, Ha1998} and also Section~\ref{ssec:Symplectic_GWPD} below.

Recently, inspired by the work of \citet{Lu2008} and \citet{FaLu2006}, \citet{OhLe2013} described the (reduced) dynamics of the Gaussian wave packet~\eqref{eq:chi} as a Hamiltonian system on $\R^{2d} \times \mathbb{H}_{d}$ (as opposed to just $\R^{2d}$):
One has a symplectic structure on $\R^{2d} \times \mathbb{H}_{d}$ that is naturally induced from the full Schr\"odinger dynamics as well as a Hamiltonian function on $\R^{2d} \times \mathbb{H}_{d}$ given as the expectation value of the Hamiltonian operator $\hat{H}$ with respect to the Gaussian wave packet.

Upon normalization, \eqref{eq:chi} becomes
\begin{equation}
  \label{eq:psi_0}
  \psi_{0}(x) \defeq \frac{\chi_{0}(x)}{\norm{\chi_{0}}} = \parentheses{ \frac{\det\mathcal{B}}{(\pi\hbar)^{d}} }^{1/4} \exp\braces{ \frac{{\rm i}}{\hbar}\brackets{ \frac{1}{2}(x - q)^{T}(\mathcal{A} + {\rm i}\mathcal{B})(x - q) + p \cdot (x - q) + \phi } },
\end{equation}
and its Wigner transform---called the {\em Gaussian state Wigner function} throughout the paper---is also a Gaussian defined on the phase space or the cotangent bundle $T^{*}\R^{d} \cong \R^{2d}$:
\begin{align}
  \label{eq:Wigner-psi_0}
  \mathcal{W}_{\psi_{0}}(\zeta)
  &\defeq \int_{\R^{d}} \mathrm{e}^{-\frac{\rmi}{\hbar} w \cdot y}\, \overline{\psi_{0}(x - y/2)}\, \psi_{0}(x + y/2)\,\mathrm{d}y \nonumber\\
  &= \frac{1}{(\pi\hbar)^{d}} \exp\brackets{ -\frac{1}{\hbar}(\zeta - z)^{T} \sigma(\mathcal{C})^{-1} (\zeta - z) },
\end{align}
where $\zeta \defeq (x,w)$ and $z \defeq (q,p)$ are both in $\R^{2d}$ and $\sigma\colon \mathbb{H}_{d} \to \sym(2d,\R)$ is the covariance matrix defined as
\begin{equation}
  \label{eq:sigma}
  \sigma(\mathcal{C}) \defeq 
  \begin{bmatrix}
    \mathcal{B}^{-1} & \mathcal{B}^{-1} \mathcal{A} \\
    \mathcal{A} \mathcal{B}^{-1} & \mathcal{A} \mathcal{B}^{-1} \mathcal{A} + \mathcal{B}
  \end{bmatrix}.
\end{equation}

Recently, \citet{BoTr2016} discovered a non-canonical Poisson bracket that describes the dynamics of the Gaussian Wigner function\footnote{Note that Gaussian Wigner functions are not always Wigner transforms of Gaussian wave packets: Indeed, Gaussian Wigner functions may describe mixed and pure quantum states depending on the form of the covariance matrix (pure Gaussian states are Gaussian wave packets). As shown by \citet{Li1986}, a Gaussian Wigner function whose covariance matrix is symplectic identifies the Wigner transform of a Gaussian wave packet, while more general forms of the covariance matrix identify mixed Gaussian states \cite{GrVa1988,SiSuMu1987}.}
\begin{equation}
  \label{eq:W_0}
  \mathcal{W}_{0}(\zeta)
  = \frac{1}{(\pi\hbar)^{d}\sqrt{\det \Sigma}} \exp\brackets{ -\frac{1}{\hbar}(\zeta - z)^{T} \Sigma^{-1} (\zeta - z) },
\end{equation}
as a Hamiltonian system with the Hamiltonian function $h(z,\Sigma)$ given by the expectation value
\begin{equation*}
  h(z,\Sigma) = \int \mathcal{W}_0(\zeta)H_{\mathcal{W}}(\zeta)\,\mathrm{d}\zeta,
\end{equation*}
where $H_{\mathcal{W}}(\zeta)$ is the Weyl symbol of the Hamiltonian operator $\hat{H}$.

These recent works \cite{OhLe2013} and \cite{BoTr2016} shed a new light on the dynamics of the Gaussian wave packet and the Gaussian Wigner function.
In fact, these Hamiltonian formulations yield a slightly different form of equations for the phase space variable $z = (q,p)$ from those conventional results in the earlier literature mentioned above:
The symplectic Gaussian wave packet dynamics in \cite{OhLe2013} yields a correction force term in the evolution equation for the momentum $p$ (see \eqref{eq:OhLe-qp} below), and the Hamiltonian dynamics of the Gaussian Wigner function in \cite{BoTr2016} also possesses a similar property.
To put it differently, in the conventional work, the phase space variables $z = (q,p)$ evolves according to a classical Hamiltonian system and is decoupled from the dynamics of $\mathcal{A} + {\rm i}\mathcal{B}$ or $\Sigma$; as a result, the entire system is {\em not} Hamiltonian.
In contrast, \cite{OhLe2013} and \cite{BoTr2016} recast the systems for $(q,p,\mathcal{A},\mathcal{B})$ and $(q,p,\Sigma)$, respectively, as Hamiltonian systems along with the natural symplectic and Poisson structures and Hamiltonians.
These formulations naturally give rise to correction terms as a result of the coupling.

Our main motivation is to unfold the geometry behind the relationship between the Hamiltonian dynamics systems for the variables $(q,p,\mathcal{A},\mathcal{B})$ and $(q,p,\Sigma)$.
Given that both systems are Hamiltonian and require modifications of the conventional picture, it is natural to expect a symplectic/Poisson-geometric connection between them.

\subsection{Main Results and Outline}
This paper exploits ideas from symplectic geometry to build a bridge between the above-mentioned recent works \cite{OhLe2013} and \cite{BoTr2016} on the Gaussian wave packet~\eqref{eq:psi_0} and its Wigner transform~\eqref{eq:Wigner-psi_0}.
The main result, Theorem~\ref{thm:Kostant-Siegel}, states that the momentum map corresponding to the natural action of the symplectic group $\Sp(\R^{2d})$ on the Siegel upper half space $\mathbb{H}_{d}$ gives the covariance matrix~\eqref{eq:sigma} of the Gaussian state Wigner function.
Its consequence, summarized in Corollary~\ref{cor:collective_dynamics} in Section~\ref{sec:collective_dynamics}, is that the dynamics of the covariance matrix---under quadratic potentials---is a collective dynamics, and is hence given by the Lie--Poisson equation on the coadjoint orbits in $\sp(\R^{2d})^{*}$.
Finally, Section~\ref{sec:DynamicsOfSCSWigner} generalizes this result to non-quadratic potentials by relating the geometry and dynamics of the Gaussian wave packets with those of the Gaussian state Wigner functions.
Particularly, Proposition~\ref{prop:Kostant-GWP} relates the symplectic structure (and the Poisson bracket) found in \cite{OhLe2013} with the Poisson bracket found in \cite{BoTr2016}, thereby establishing a geometric link between the two formulations.
We also numerically demonstrate that our dynamics gives a better approximation to the dynamics of expectation values than the classical solutions do.

\section{Hamiltonian Dynamics of Gaussian Wave Packet and Gaussian Wigner Function}
This section gives a brief review of the works~\cite{OhLe2013} and \cite{BoTr2016} mentioned above.

\subsection{Symplectic Structure and Gaussian Wave Packets}
\label{ssec:Symplectic_GWPD}
It is well known that the Gaussian wave packet~\eqref{eq:chi} gives an exact solution to the time-dependent Schr\"odinger equation
\begin{equation}
  \label{eq:Schroedinger}
  {\rm i}\hbar\,\pd{}{t}\psi(t,x) = -\frac{\hbar^{2}}{2m} \Delta \psi(t,x) + V(x)\,\psi(t,x)
\end{equation}
with quadratic potential $V$ if the parameters $(q, p, \mathcal{A}, \mathcal{B}, \phi, \delta) \in \R^{2d} \times \mathbb{H}_{d} \times \mathbb{S}^{1} \times \R$ evolve in time according to a set of ODEs; see, e.g., \citet{He1975a,He1976b} and \citet{Ha1980, Ha1981, Hagedorn1985, Ha1998}.
This set of ODEs is the classical Hamiltonian system 
\begin{equation*}
  \dot{q} = \frac{p}{m},
  \qquad
  \dot{p} = -\pd{V}{q}(q)
\end{equation*}
coupled with additional equations for the rest of the variables $(\mathcal{A}, \mathcal{B}, \phi, \delta)$.

The idea of reformulating this whole set of ODEs for $(q, p, \mathcal{A}, \mathcal{B}, \phi, \delta)$ as a Hamiltonian system has been around for quite a while; see, e.g., \citet{PaSc1994}, \citet{FaLu2006}, and \citet[Section~II.4]{Lu2008}.
\citet{OhLe2013} built on these works from the symplectic-geometric point of view and came up with a Hamiltonian system on $\R^{2d} \times \mathbb{H}_{d} \times \mathbb{S}^{1} \times \R$ with an $\mathbb{S}^{1}$ phase symmetry, and by applying the Marsden--Weinstein reduction~\cite{MaWe1974} (see also \citet[Sections~1.1 and 1.2]{MaMiOrPeRa2007}), obtained the Hamiltonian system
\begin{equation}
  \label{eq:HamiltonianSystem-GWP}
  \ins{X_{H}}\Omega = \d{H}
\end{equation}
on the reduced symplectic manifold $\mathcal{P} \defeq \R^{2d} \times \mathbb{H}_{d} = \{ (q, p, \mathcal{A}, \mathcal{B}) \}$ that is equipped with the symplectic form
\begin{equation}
  \label{eq:Omega}
  \Omega \defeq \d{q^{j}} \wedge \d{p_{j}} + \frac{\hbar}{4}\, \d\mathcal{B}^{-1}_{kl} \wedge \d\mathcal{A}_{kl}.
\end{equation}
Note that we use Einstein's summation convention throughout the paper unless otherwise stated.
Given a Hamiltonian $H\colon \mathcal{P} \to \R$, \eqref{eq:HamiltonianSystem-GWP} determines the vector field $X_{H}$ on $\mathcal{P}$; in coordinates it is written as
\begin{equation*}
  \dot{q} = \pd{H}{p},
  \qquad
  \dot{p} = -\pd{H}{q},
  \qquad
  \dot{\mathcal{A}} = \frac{4}{\hbar}\mathcal{B} \pd{H}{\mathcal{B}} \mathcal{B},
  \qquad
  \dot{\mathcal{B}} = -\frac{4}{\hbar}\mathcal{B} \pd{H}{\mathcal{A}} \mathcal{B},
\end{equation*}
where $\tpd{H}{\mathcal{B}}$ stands for the matrix whose $(m,n)$-entry is $\tpd{H}{\mathcal{B}_{mn}}$.
In our setting, it is natural to select the Hamiltonian $H\colon \mathcal{P} \to \R$ as
\begin{equation}
  \label{eq:H}
  H = \frac{p^{2}}{2m} + V(q)
  + \frac{\hbar}{4}\tr\brackets{ \mathcal{B}^{-1}\parentheses{ \frac{\mathcal{A}^{2} + \mathcal{B}^{2}}{m} + D^{2}V(q) } },
\end{equation}
where $D^{2}V$ denotes the Hessian matrix of the potential $V$.
In fact, it is an $O(\hbar^{2})$ approximation to the expectation value $\bigl\langle \psi_{0}, \hat{H}\psi_{0} \bigr\rangle$ of the Hamiltonian operator $\hat{H}$.
Then we have
\begin{subequations}
  \label{eq:OhLe}
  \begin{gather}
    \label{eq:OhLe-qp}
    \dot{q} = \frac{p}{m},
    \qquad
    \dot{p} = -\pd{}{q}\brackets{ V(q) + \frac{\hbar}{4} \tr\parentheses{ \mathcal{B}^{-1} D^{2}V(q) } },
    \medskip\\
    \label{eq:OhLe-AB}
    \dot{\mathcal{A}} = -\frac{1}{m}(\mathcal{A}^{2} - \mathcal{B}^{2}) - D^{2}V(q),
    \qquad
    \dot{\mathcal{B}} = -\frac{1}{m}(\mathcal{A}\mathcal{B} + \mathcal{B}\mathcal{A}).
  \end{gather}
\end{subequations}
This equation differs from those of \citet{He1975a,He1976b} and \citet{Ha1980, Ha1981, Hagedorn1985, Ha1998} by the $O(\hbar)$ correction term to the potential in the second equation; see \citet{OhLe2013} and \cite{Oh2015b} for the effects of this correction term.

The corresponding Poisson structure $\PB{\,\cdot\,}{\,\cdot\,}_{\mathcal{P}}$ on $\mathcal{P}$ is defined as follows:
For any $F, G \in C^{\infty}(\mathcal{P})$, let $X_{F}$ and $X_{G}$ be the corresponding Hamiltonian vector fields, i.e., $\ins{X_{F}}\Omega = \d{F}$ and similarly for $X_{G}$, then
\begin{equation}
  \label{eq:PB-GWP}
  \PB{F}{G}_{\mathcal{P}} \defeq \Omega(X_{F},X_{G})
  = \PB{F}{G}_{\R^{2d}}
  - \frac{4}{\hbar}\PB{F}{G}_{\mathbb{H}_{d}},
\end{equation}
with
\begin{gather}
  \PB{F}{G}_{\R^{2d}} \defeq \pd{F}{q^{i}} \pd{G}{p_{i}} - \pd{G}{q^{i}} \pd{F}{p_{i}},
  \\
  \PB{F}{G}_{\mathbb{H}_{d}} \defeq
  -\parentheses{
    \pd{F}{\mathcal{B}^{-1}_{jk}} \pd{G}{\mathcal{A}_{jk}} - \pd{G}{\mathcal{B}^{-1}_{jk}} \pd{F}{\mathcal{A}_{jk}}
  },
  \label{eq:PB-H_d}
\end{gather}
where each of the brackets is calculated by holding the remaining variables (that are not involved in the bracket) fixed; we employ this convention throughout the paper to simplify the notation.

\subsection{Lie--Poisson Structure for Gaussian Moments}
The center $z$ and the covariance matrix $\Sigma$ (assumed to be positive definite) of the Gaussian Wigner function $\mathcal{W}_{0}$ from \eqref{eq:W_0} are given in terms of the first two moments of the Gaussian Wigner function \eqref{eq:Wigner-psi_0}, that is\footnote{Notice that, although here we call $\Sigma$ ``covariance matrix'', this differs from the usual definition in statistics (given by ${\exval{\zeta \otimes \zeta}_{0}}- {\exval{\zeta}_{0} \otimes \exval{\zeta}_{0}}$) by an irrelevant multiplicative factor $\hbar/2$.}
\begin{equation}
  \label{eq:moments}
  \exval{\zeta}_{0} = z,
  \qquad
  \frac{1}{2}\parentheses{
    \exval{\zeta \otimes \zeta}_{0}- \exval{\zeta}_{0} \otimes \exval{\zeta}_{0}
  } = \frac{\hbar}{4}\Sigma,
\end{equation}
where we have used the following expectation value notation with respect to $\mathcal{W}_{0}$ as well as more general Wigner function $\mathcal{W}$:
For an observable $a\colon \R^{2d} \to \R$,
\begin{equation}
  \exval{a}_{0} \defeq \int\mathcal{W}_{0}(\zeta) a(\zeta)\,\mathrm{d}\zeta,
  \qquad
  \exval{a} \defeq \int\mathcal{W}(\zeta) a(\zeta)\,\mathrm{d}\zeta.
  \label{eq:exval}
\end{equation}

In \cite{BoTr2016}, the first two moments of the Wigner quasiprobability density $\mathcal{W}(\zeta)$ were characterized as the momentum map corresponding to the action
\begin{equation}
  \label{eq:Jac_action}
  \big({(S,{\sf z},\varphi)}\cdot\mathcal{W}\big)(\zeta)=\mathcal{W}(S\zeta+{\sf z})
  \,,\qquad
  \text{\ with \ }
  \qquad
  S\in \Sp(\R^{2d})\,,\quad 
  ({\sf z},\varphi)\in\mathsf{H}(\R^{2d})\,,
\end{equation}
of the \emph{Jacobi group} $\mathsf{Jac}(\R^{2d}) \defeq \Sp(\R^{2d})\,\circledS\, \mathsf{H}(\R^{2d})$, i.e., the semidirect product of the symplectic group $\Sp(\R^{2d})$ and the Heisenberg group $\mathsf{H}(\R^{2d})$.
This group structure has attracted some attention over the years mainly because of its relation to squeezed coherent states \cite{Be2007} and, more recently, because of its connections to certain integrable geodesic flows on the symplectic group \cite{BlBrIsMaRa2009, HoTr2009}.
Here, the space of quasiprobability densities is equipped with a Poisson structure given by the following Lie--Poisson bracket on the space $\mathfrak{s}^{*}$ of the set of Wigner functions \cite{BiMo1991} (see Appendix~\ref{appendix} for more details):
\begin{equation*}
\{F,K\}_{\mathfrak{s}^{*}}(\mathcal{W})=\int_{\R^{2d}}\!\mathcal{W}(\zeta)\left\{\!\!\left\{\frac{\delta F}{\delta \mathcal{W}},\frac{\delta K}{\delta \mathcal{W}}\right\}\!\!\right\}(\zeta)\,\mathrm{d}\zeta
\,,
\end{equation*}
where $\{\!\{\cdot,\cdot\}\!\}$ denotes the Moyal bracket~\cite{Gr1946, Mo1949}.
In addition, the symplectic group is defined as follows:
\begin{equation*}
  \Sp(\R^{2d}) \defeq
  \setdef{
    S \in \Mat(2d,\R)
  }{ S^{T} \mathbb{J} S = \mathbb{J}}
  \quad\text{with}\quad
  \mathbb{J} =
  \begin{bmatrix}
    0 & I_{d} \\
    -I_{d} & 0
  \end{bmatrix},
\end{equation*}
whereas the Heisenberg group
\begin{equation*}
   \mathsf{H}(\R^{2d}) = \setdef{(\mathsf{z},\varphi) \in \R^{2d+1}}{ \mathsf{z} = (\mathsf{q},\mathsf{p}) \in \R^{2d},\, \varphi \in \R}  
\end{equation*}
is equipped with the multiplication rule
\begin{equation*}
  (\mathsf{z}_{1}, \theta_{1}) \cdot (\mathsf{z}_{2}, \theta_{2})
  = \parentheses{
    \mathsf{z}_{1} + \mathsf{z}_{2}, \theta_{1} + \theta_{2} - \Omega_{\R^{2d}}(\mathsf{z}_{1}, \mathsf{z}_{2})/2
  },
\end{equation*}
where $\Omega_{\R^{2d}}$ is the standard symplectic form on $\R^{2d}$, i.e., setting $\mathsf{z}_{1} = (\mathsf{q}_{1}, \mathsf{p}_{1})$ and $\mathsf{z}_{2} = (\mathsf{q}_{2}, \mathsf{p}_{2})$,
\begin{equation*}
  \Omega_{\R^{2d}}(\mathsf{z}_{1}, \mathsf{z}_{2}) =\mathsf{z}_{1}\cdot \mathbb{J} \mathsf{z}_{2} = \mathsf{q}_{1} \cdot \mathsf{p}_{2} - \mathsf{q}_{2} \cdot \mathsf{p}_{1}.
\end{equation*}
The semidirect product $\Sp(\R^{2d})\,\circledS\, \mathsf{H}(\R^{2d})$ is defined in terms of the natural $\Sp(\R^{2d})$-action on $\mathsf{H}(\R^{2d})$, i.e., $(\mathsf{z},\theta)\mapsto (S\mathsf{z},\theta)$ with $S\in \Sp(\R^{2d})$; as a result, the group multiplication for the Jacobi group is given by
\[
(S,\mathsf{z},\theta)(S',\mathsf{z}',\theta')=(SS',\mathsf{z}+S\mathsf{z}',\theta+\theta' - \Omega_{\R^{2d}}(\mathsf{z}, S\mathsf{z}')/2)
\,.
\]

In the context of Gaussian quantum states, the Jacobi group plays exactly the same role as in classical Liouville (Vlasov) dynamics \cite{GaTr2012}, so that the momentum map structure of the first two Wigner moments has an identical correspondent in the classical case (one simply replaces the Wigner function by a Liouville distribution).
More specifically, the momentum map $\mathbf{J}_{\mathfrak{s}^{*}}\colon \mathfrak{s}^{*} \to \mathfrak{jac}(\R^{2d})^{*}$ corresponding to the action~\eqref{eq:Jac_action} is (see Appendix~\ref{appendix} for a verification)
\begin{equation}
  \label{eq:J_sstar}
  \mathbf{J}_{\mathfrak{s}^{*}}(\mathcal{W}) = \left(\frac{1}2\Bbb{J}^{T}\langle\zeta\otimes\zeta\rangle,\,\Bbb{J}^{T}\langle\zeta\rangle, \,\langle 1\rangle\right).
\end{equation}
Here, $\mathfrak{jac}(\Bbb{R}^{2d})^*$ is the dual of the Lie algebra $\mathfrak{jac}(\Bbb{R}^{2d}) \defeq \mathfrak{sp}(\Bbb{R}^{2d})\,\circledS\, \mathfrak{h}(\Bbb{R}^{2d})$ of $\mathsf{Jac}(\R^{2d})$, with $\mathfrak{sp}(\Bbb{R}^{2d})$ and $\mathfrak{h}(\Bbb{R}^{2d})$ being the Lie algebras of $\Sp(\R^{2d})$ and $\mathsf{H}(\R^{2d})$, respectively.
The momentum map $\mathbf{J}_{\mathfrak{s}^{*}}$ is equivariant and hence is Poisson (see, e.g., \citet[Theorem~12.4.1]{MaRa1999}) with respect to the above $\PB{\,\cdot\,}{\,\cdot\,}_{\mathfrak{s}^{*}}$ and the $(-)$-Lie--Poisson bracket $\PB{\,\cdot\,}{\,\cdot\,}^{-}_{\mathfrak{jac}(\R^{2d})^{*}}$ on $\mathfrak{jac}(\R^{2d})^{*} = \{ (\Pi, \lambda, \alpha) \}$ defined as follows: For any $f, g \in C^{\infty}(\mathfrak{jac}(\R^{2d})^{*})$,
\begin{equation}
  \label{eq:PB-jac}
  \PB{f}{g}_{\mathfrak{jac}(\R^{2d})^{*}}^{-}(\Pi, \lambda, \alpha) \defeq 
  \alpha \PB{f}{g}_{\R^{2d}}
  - \lambda \cdot \parentheses{ \Fd{f}{\Pi} \pd{g}{\lambda} - \Fd{g}{\Pi} \pd{f}{\lambda} }
  - \tr\parentheses{ \Pi^{T} \brackets{ \Fd{f}{\Pi}, \Fd{g}{\Pi} }_{\mathfrak{sp}} },
\end{equation}
where $[\,\cdot\,,\,\cdot\,]_{\mathfrak{sp}}$ is the standard commutator on $\mathfrak{sp}(\R^{2d})$; we also identified $(\Bbb{R}^{2d})^*\simeq\Bbb{R}^{2d}$ via the usual dot product.
The differential ${\delta f}/{\delta\Pi} \in \sp(\R^{2d})$ is defined in terms of the natural dual pairing $\ip{\,\cdot\,}{\,\cdot\,}_{\mathfrak{sp}}\colon \sp(\R^{2d})^{*} \times \sp(\R^{2d}) \to \R$ as follows:
For any $\Delta\Pi \in \sp(\R^{2d})^{*}$
\begin{equation*}
  \left.\od{}{\varepsilon} f(\Pi+\varepsilon\Delta\Pi) \right|_{\varepsilon=0}
  = \ip{\Delta\Pi}{\Fd{f}{\Pi}}_{\mathfrak{sp}} = \tr\parentheses{\Delta\Pi^{T} \Fd{f}{\Pi} },
\end{equation*}
where we identified $\sp(\R^{2d})^{*}$ with $\sp(\R^{2d})$ via the inner product on $\mathfrak{sp}(\R^{2d})$ defined in \eqref{eq:ip-sp} below.
The other differentials denoted with $\delta$ are defined similarly.

Notice however that the image of the momentum map $\mathbf{J}_{\mathfrak{s}^{*}}$ in \eqref{eq:J_sstar} is not quite identified with those moments of interest from \eqref{eq:moments}; in other words, $\mathfrak{jac}(\R^{2d})^{*} = \mathfrak{sp}(\Bbb{R}^{2d})^{*} \times \mathfrak{h}(\Bbb{R}^{2d})^{*}$ is not a natural space in which those moments live.
However, by exploiting the ``untangling map'' of \citet[Proposition 2.2]{KrMa1987} and the identification of $\sp(\R^{2d})^{*}$ with $\sym(2d,\R)$ outlined in Section~\ref{ssec:sp-sym}, the Lie--Poisson bracket~\eqref{eq:PB-jac} gives rise to the Poisson bracket
\begin{equation}
  \label{eq:PB-GaussianWigner}
  \PB{f}{g}(z,\Sigma) =
  \PB{f}{g}_{\R^{2d}}
  - \frac{4}{\hbar}\tr\parentheses{ \Sigma\! \left[ \Fd{f}{\Sigma}, \Fd{g}{\Sigma} \right]_{\sym} }
\end{equation}
on $\R^{2d} \times \sym(2d,\R) = \{(z, \Sigma)\}$; this space is naturally identified as the space of the Gaussian moments $(z,\Sigma)$.
See Appendix~\ref{sec:PB-moments} for the details of the derivation of the above Poisson bracket.
With a Hamiltonian $h\colon \R^{2d} \times \sym(2d,\R) \to \R$, we have
\begin{equation}
  \dot{z} = \PB{z}{h}_{\R^{2d}}
  \qquad
  \dot{\Sigma} = \frac{4}{\hbar} \left( \mathbb{J}\Fd{h}{\Sigma}\Sigma - \Sigma\Fd{h}{\Sigma}\mathbb{J} \right),
  \label{varianceformulation}
\end{equation}
which are equivalent to (5.2) in \citet{BoTr2016} (up to a sign misprint therein).

If the classical Hamiltonian is quadratic, then \eqref{varianceformulation} with $\Sigma = \sigma(\mathcal{C})$ describes the time evolution of the Gaussian state Wigner function~\eqref{eq:Wigner-psi_0} corresponding to the dynamics of the Gaussian wave packet~\eqref{eq:psi_0} as an exact solution to the corresponding Schr\"odinger equation.
See Section~\ref{sec:DynamicsOfSCSWigner} for the dynamics with non-quadratic Hamiltonians.

For Hamiltonians $h(z, \Sigma)$ that are linear in $\Sigma$, these equations recover the dynamics (12) and (13) in \cite{GrSc2011} (suitably specialized to Hermitian quantum mechanics).
However, certain approximate models in chemical physics \cite{PrPe2000} make use of  nonlinear terms in $\Sigma$, as they are obtained by Gaussian moment closures of the type $\langle\zeta_i\zeta_j\zeta_k\rangle=\langle\zeta_i\rangle\langle\zeta_j\zeta_k\rangle+\langle\zeta_i\rangle\langle\zeta_j\rangle\langle\zeta_k\rangle$.
One may perform such closures in the expression of the total energy to obtain the Hamiltonian of the form $h(z,\Sigma)$, and then can formulate, along with the Poisson bracket~\eqref{eq:PB-GaussianWigner}, the dynamics of $(z, \Sigma)$ as a Hamiltonian system.

\section{Covariance Matrix as a Momentum Map}
Our goal is to establish a link between the symplectic structure~\eqref{eq:Omega} or the Poisson bracket~\eqref{eq:PB-GWP} on $\R^{2d} \times \mathbb{H}_{d}$ and the Poisson bracket~\eqref{eq:PB-GaussianWigner} on $\R^{2d} \times \sym(2d,\R)$.
In this section, we focus on the correspondence between the second parts---$\mathbb{H}_{d}$ and $\sym(2d,\R)$---of these constituents.
The main result, Theorem~\ref{thm:Kostant-Siegel} below, states that this link is made via the momentum map corresponding to the natural action of the symplectic group $\Sp(\R^{2d})$ on the Siegel upper half space $\mathbb{H}_{d}$.
We start off with a brief review of the geometry of $\mathbb{H}_{d}$ in Section~\ref{ssec:Siegel}, and then after giving a brief account of the identification between $\sp(\R^{2d})^{*}$ and $\sym(2d,\R)$ alluded above in Section~\ref{ssec:sp-sym}, we state and prove the main result that the momentum map yields the covariance matrix~\eqref{eq:sigma} in Section~\ref{ssec:MomentumMap}.

\subsection{Geometry of the Siegel Upper Half Space}
\label{ssec:Siegel}
It is well known that the Siegel upper half space $\mathbb{H}_{d}$ defined in \eqref{eq:H_d} is a homogeneous space; more specifically, we can show that
\begin{equation*}
  \mathbb{H}_{d} \cong \Sp(\R^{2d})/\U(d),
\end{equation*}
where $\U(d)$ is the unitary group of degree $d$; see \citet{Si1943} and also \citet[Section~4.5]{Fo1989} and \citet[Exercise~2.28 on p.~48]{McSa1999}.
To see this, let us first rewrite the definition of $\Sp(\R^{2d})$ using block matrices consisting of $d \times d$ submatrices, i.e.,
\begin{equation}
  \label{def:Sp2d-block}
  \Sp(\R^{2d}) \defeq
  \setdef{
    \begin{bmatrix}
      S_{11} & S_{12} \\
      S_{21} & S_{22}
    \end{bmatrix}
    \in \Mat(2d,\R)
  }{
    \begin{array}{c}
      S_{11}^{T}S_{21} = S_{21}^{T}S_{11},\, S_{12}^{T}S_{22} = S_{22}^{T}S_{12}, \smallskip\\
      S_{11}^{T}S_{22} - S_{21}^{T}S_{12} = I_{d}
    \end{array}
  },
\end{equation}
and define the (left) action $\Phi$ of $\Sp(\R^{2d})$ on $\mathbb{H}_{d}$ by the generalized linear fractional transformation
\begin{equation}
  \label{eq:Sp_action}
  \Phi_{(\,\cdot\,)}\colon \Sp(\R^{2d}) \times \mathbb{H}_{d} \to \mathbb{H}_{d};
  \quad
  \parentheses{
    \begin{bmatrix}
      S_{11} & S_{12} \\
      S_{21} & S_{22}
    \end{bmatrix},
    \mathcal{C} = \mathcal{A} + \rmi\mathcal{B}
  }
  \mapsto
  (S_{21} + S_{22}\mathcal{C})(S_{11} + S_{12}\mathcal{C})^{-1}.
\end{equation}
This action is transitive: By choosing
\begin{equation}
  \label{eq:S-mathcalAB}
  S =
  \Xi(\mathcal{A},\mathcal{B})
  \defeq
  \begin{bmatrix}
    I_{d} & 0 \\
    \mathcal{A} & I_{d}
  \end{bmatrix}
  \begin{bmatrix}
    \mathcal{B}^{-1/2} & 0 \\
    0 & \mathcal{B}^{1/2}
  \end{bmatrix}
  =
  \begin{bmatrix}
    \mathcal{B}^{-1/2} & 0 \\
    \mathcal{A}\mathcal{B}^{-1/2} & \mathcal{B}^{1/2}
  \end{bmatrix},
\end{equation}
which is easily shown to be symplectic, we have
\begin{equation}
  \label{eq:S-mathcalAB2}
  \Phi_{\Xi(\mathcal{A},\mathcal{B})}({\rm i}I_{d}) = \mathcal{A} + {\rm i}\mathcal{B}.
\end{equation}
The isotropy subgroup of the element ${\rm i}I_{d} \in \mathbb{H}_{d}$ is given by
\begin{align*}
  \Sp(\R^{2d})_{{\rm i}I_{d}} &= \setdef{
    \begin{bmatrix}
      U  & V \\
      -V & U
    \end{bmatrix} \in \Mat(2d,\R)
  }{U^{T}U + V^{T}V = I_{d},\, U^{T}V = V^{T}U}
  \nonumber\\
  &= \Sp(\R^{2d}) \cap \Orth(2d),
\end{align*}
where $\Orth(2d)$ is the orthogonal group of degree $2d$; however $\Sp(\R^{2d}) \cap \Orth(2d)$ is identified with $\U(d)$ as follows:
\begin{equation*}
  \Sp(\R^{2d}) \cap \Orth(2d) \to \U(d);
  \quad
  \begin{bmatrix}
    U  & V \\
    -V & U
  \end{bmatrix}
  \mapsto U + {\rm i}V.
\end{equation*}
Hence $\Sp(\R^{2d})_{{\rm i}I_{d}} \cong \U(d)$ and thus $\mathbb{H}_{d} \cong \Sp(\R^{2d})/\U(d)$.
We may then construct the corresponding quotient map as follows:
\begin{equation*}
  \pi_{\U(d)}\colon \Sp(\R^{2d}) \to \Sp(\R^{2d})/\U(d) \cong \mathbb{H}_{d};
  \quad
  Y \mapsto \Phi_{Y}({\rm i}I_{d}),
\end{equation*}
or more explicitly,
\begin{equation*}
  \pi_{\U(d)}\parentheses{
    \begin{bmatrix}
      S_{11} & S_{12} \\
      S_{21} & S_{22}
    \end{bmatrix}
  }
  = (S_{21} + {\rm i}S_{22})(S_{11} + {\rm i}S_{12})^{-1},
\end{equation*}
where $S_{11} + {\rm i}S_{12}$ can be shown to be invertible if $
\begin{tbmatrix}
  S_{11} & S_{12} \\
  S_{21} & S_{22}
\end{tbmatrix} \in \Sp(\R^{2d})$.
Let $L_{S}\colon \Sp(\R^{2d}) \to \Sp(\R^{2d})$ be the left multiplication by $S \in \Sp(\R^{2d})$, i.e., $L_{S}(Y) = S Y$ for any $Y \in \Sp(\R^{2d})$.
Then it is easy to see that
\begin{equation}
  \label{eq:Phi-pi}
  \Phi_{S} \circ \pi_{\U(d)} = \pi_{\U(d)} \circ L_{S}
\end{equation}
or the diagram below commutes, i.e., $\Phi$ defined in \eqref{eq:Sp_action} is in fact a left action.
\begin{equation*}
  \begin{tikzcd}[row sep=7ex, column sep=7ex]
    \Sp(\R^{2d}) \arrow{r}{L_{S}} \arrow{d}[swap]{\pi_{\U(d)}} & \Sp(\R^{2d}) \arrow{d}{\pi_{\U(d)}} \\
    \mathbb{H}_{d} \arrow{r}[swap]{\Phi_{S}} &\mathbb{H}_{d}
  \end{tikzcd}
\end{equation*}

The Siegel upper half space $\mathbb{H}_{d}$ is also a symplectic manifold with symplectic form~(see \citet{Si1943} and also \citet{Oh2015c})
\begin{equation}
  \label{eq:Omega-H_d}
  \Omega_{\mathbb{H}_{d}}
  \defeq \mathcal{B}^{-1}_{lj} \mathcal{B}^{-1}_{km} \d\mathcal{B}_{lm} \wedge \d\mathcal{A}_{jk}
  = -\d\mathcal{B}^{-1}_{jk} \wedge \d\mathcal{A}_{jk}.
\end{equation}
In fact, one may define the canonical one-form $\Theta_{\mathbb{H}_{d}}$ on $\mathbb{H}_{d}$ as
\begin{equation}
  \label{eq:Theta-H_d}
  \Theta_{\mathbb{H}_{d}} \defeq -\tr(\mathcal{A}\, \d\mathcal{B}^{-1})
\end{equation}
so that $\Omega_{\mathbb{H}_{d}} = -\d\Theta_{\mathbb{H}_{d}}$, and the corresponding Poisson bracket is $\PB{\,\cdot\,}{\,\cdot\,}_{\mathbb{H}_{d}}$ shown in \eqref{eq:PB-H_d}.

\subsection{Symplectic Algebra and Lie Algebra of Symmetric Matrices}
\label{ssec:sp-sym}
The following bracket renders the space $\sym(2d,\R)$ of $2d \times 2d$ symmetric real matrices a Lie algebra:
\begin{equation}
  \label{eq:Lie_bracket-Sym}
  [\,\cdot\,, \,\cdot\,]_{\sym}\colon \sym(2d,\R) \times \sym(2d,\R) \to \sym(2d,\R);
  \qquad
  [\xi, \eta]_{\sym} \defeq \xi \mathbb{J}^{T} \eta - \eta \mathbb{J}^{T} \xi.
\end{equation}
We then identify the symplectic algebra $\sp(\R^{2d})$ with $\sym(2d,\R)$ via the following ``tilde map'':
\begin{equation}
  \label{eq:tilde_map}
  \tilde{(\,\cdot\,)}\colon \sym(2d,\R) \to \sp(\R^{2d});
  \qquad
  \xi =
  \begin{bmatrix}
    \xi_{11} & \xi_{12} \\
    \xi_{12}^{T} & \xi_{22}
  \end{bmatrix}
  \mapsto
  \mathbb{J}^{T} \xi
  =
  \begin{bmatrix}
    -\xi_{12}^{T} & -\xi_{22} \\
    \xi_{11} & \xi_{12}
  \end{bmatrix}
  \eqdef \tilde{\xi},
\end{equation}
where $\xi_{12} \in \Mat(d,\R)$, i.e., it is a $d \times d$ real matrix, and $\xi_{11}, \xi_{22} \in \sym(d,\R)$.
So writing $\tilde{\xi} = \begin{tbmatrix}
  \tilde{\xi}_{11} & \tilde{\xi}_{12} \\
  \tilde{\xi}_{12}^{T} & \tilde{\xi}_{22}
\end{tbmatrix} \in \sp(\R^{2d})$, the identification~\eqref{eq:tilde_map} is written explicitly in terms of the block components as follows:
\begin{equation}
  \label{eq:xi-tildexi}
  \tilde{\xi}_{11} = -\xi_{12}^{T},
  \qquad
  \tilde{\xi}_{12} = -\xi_{22},
  \qquad
  \tilde{\xi}_{21} = \xi_{11},
  \qquad
  \tilde{\xi}_{22} = \xi_{12}.
\end{equation}
In fact, it is easy to see that this is a Lie algebra isomorphism:
Let $[\tilde{\xi}, \tilde{\eta}]_{\sp} \defeq \tilde{\xi} \tilde{\eta} - \tilde{\eta} \xi$ be the standard Lie bracket of $\sp(\R^{2d})$; then, for any $\xi, \eta \in \sym(2d,\R)$, 
\begin{equation*}
  [\tilde{\xi}, \tilde{\eta}]_{\sp} = \widetilde{ [\xi, \eta] }_{\sym}.
\end{equation*}
One may also define inner products on both spaces as follows:
\begin{equation}
  \label{eq:ip-sp}
  \ip{\,\cdot\,}{\,\cdot\,}_{\sp}\colon \sp(\R^{2d}) \times \sp(\R^{2d}) \to \R;
  \qquad
  (\tilde{\xi}, \tilde{\eta}) \mapsto \tr(\tilde{\xi}^{T} \tilde{\eta})
\end{equation}
and 
\begin{equation*}
  \ip{\,\cdot\,}{\,\cdot\,}\colon \sym(2d,\R) \times \sym(2d,\R) \to \R;
  \qquad
  (\xi, \eta) \mapsto \tr(\xi \eta),
\end{equation*}
and so we may identify their dual spaces with themselves.
As a result, we have
\begin{equation*}
  \sp(\R^{2d})^{*} \cong \sp(\R^{2d}) \cong \sym(2d,\R) \cong \sym(2d,\R)^{*}.
\end{equation*}
The above inner products are compatible with the identification via the tilde map~\eqref{eq:tilde_map} in the sense that $\ip{\xi}{\eta} = \tip{\tilde{\xi}}{\tilde{\eta}}_{\sp}$.
Therefore, in what follows, we exploit the tilde map identification~\eqref{eq:tilde_map} to write elements in $\sp(\R^{2d})$, $\sp(\R^{2d})^{*}$, and $\sym(2d,\R)^{*}$ as the corresponding ones in $\sym(2d,\R)$ to simplify calculations.

Recall that the symplectic group $\Sp(\R^{2d})$ acts on its Lie algebra $\sp(\R^{2d})$ via the adjoint action, i.e., for any $S \in \Sp(\R^{2d})$ and $\tilde{\xi} \in \sp(\R^{2d})$, the adjoint action $\Ad_{S}\colon \sp(\R^{2d}) \to \sp(\R^{2d})$
\begin{equation*}
  \Ad_{S} \tilde{\xi} = S \tilde{\xi} S^{-1}.
\end{equation*}
With an abuse of notation, one may define the corresponding action\footnote{Strictly speaking, this is not an adjoint action because $\xi$ is not in $\sp(\R^{2d})$, but is identified with the above adjoint action of $\Sp(\R^{2d})$. However it is natural in the sense that $\widetilde{\Ad_{S} \xi} = \Ad_{S} \tilde{\xi}$.} $\Ad_{S}\colon \sym(2d,\R) \to \sym(2d,\R)$ of $\Sp(\R^{2d})$ by
\begin{equation*}
  \Ad_{S} \xi = (S^{-1})^{T} \xi S^{-1}.
\end{equation*}
Hence the corresponding action on the dual $\sym(2d,\R)^{*} \cong \sym(2d,\R)$ is given by
\begin{equation}
  \label{eq:Adstar_on_Sym}
  \Ad^{*}_{S^{-1}} \mu = S \mu S^{T}.
\end{equation}
One then sees easily that, for any $\xi \in \sym(2d,\R)$, the corresponding $\ad_{\xi}\colon \sym(2d,\R) \to \sym(2d,\R)$ is compatible with the Lie bracket~\eqref{eq:Lie_bracket-Sym}, i.e.,
\begin{equation*}
  \ad_{\xi}\eta = \xi \mathbb{J}^{T} \eta - \eta \mathbb{J}^{T} \xi = [\xi, \eta]_{\sym},
\end{equation*}
and then its adjoint $\ad^{*}_{\xi}\colon \sym(2d,\R)^{*} \to \sym(2d,\R)^{*}$ is given by
\begin{equation}
  \label{eq:ad_star-Sym}
  \ad^{*}_{\xi}\mu = \mathbb{J} \xi \mu - \mu \xi \mathbb{J}.
\end{equation}

The coadjoint action~\eqref{eq:Adstar_on_Sym} defines the coadjoint orbit
\begin{equation*}
  \mathcal{O} \defeq \setdef{ \Ad^{*}_{S} \mu \in \sym(2d,\R)^{*} }{ S \in \Sp(\R^{2d}) }
\end{equation*}
for each $\mu \in \sym(2d,\R)^{*}$; it is well known that $\mathcal{O}$ is equipped with the following $(\pm)$-Kirillov--Kostant--Souriau (KKS) symplectic structures:
For any $\mu \in \mathcal{O}$ and any $\xi, \eta \in \sym(2d,\R)$,
\begin{equation}
  \label{eq:KKS-sym}
  \Omega_{\mathcal{O}}^{\pm}(\mu)(\ad^{*}_{\xi}\mu, \ad^{*}_{\eta}\mu) \defeq \pm\ip{ \mu }{ [\xi,\eta]_{\sym} }
  =\pm \tr(\mu [\xi,\eta]_{\sym}).
\end{equation}
The $(\pm)$-Lie--Poisson structure on $\sym(2d,\R)^{*} \cong \sym(2d,\R)$ that is compatible with the above KKS symplectic form is given by
\begin{equation}
  \label{eq:LPB-sym}
  \PB{f}{g}_{\sym(2d,\R)}^{\pm}(\mu)
  \defeq \pm\ip{ \mu }{
    \brackets{
      \frac{\delta f}{\delta \mu}, \frac{\delta g}{\delta \mu}
    }_{\sym}
  }
  = \pm\tr\parentheses{
    \mu \brackets{
      \frac{\delta f}{\delta \mu}, \frac{\delta g}{\delta \mu}
    }_{\sym}
  }.
\end{equation}

\subsection{Momentum Map on the Siegel Upper Half Space}
\label{ssec:MomentumMap}
Recall that the symplectic group $\Sp(\R^{2d})$ acts on the Siegel upper half space $\mathbb{H}_{d}$ transitively by the action $\Phi$ shown in \eqref{eq:Sp_action}.
The main ingredient of the paper is the momentum map
\begin{equation}
  \label{def:J}
  \mathbf{J}\colon \mathbb{H}_{d} \to \sp(\R^{2d})^{*} \cong \sym(2d,\R)
\end{equation}
corresponding to this action:
Let $\xi \in \sym(2d,\R) \cong \sp(\R^{2d})$ and $\xi_{\mathbb{H}_{d}}$ be its infinitesimal generator (recall that we identify $\xi \in \sym(2d,\R)$ with $\tilde{\xi} \in \sp(\R^{2d})$), i.e.,
\begin{equation*}
  \xi_{\mathbb{H}_{d}}(\mathcal{C}) \defeq \left.\od{}{\varepsilon}\Phi_{\exp(\varepsilon\tilde{\xi})}(\mathcal{C}) \right|_{\varepsilon=0}.
\end{equation*}
Then $\mathbf{J}$ is characterized by
\begin{equation}
  \label{eq:J-HamSys}
  \ins{\xi_{\mathbb{H}_{d}}}\Omega_{\mathbb{H}_{d}} = \d\ip{\mathbf{J}(\,\cdot\,)}{\xi}
\end{equation}
for any $\xi \in \sym(2d,\R)$.
\begin{remark}
  We note in passing that the canonical one-form $\Theta_{\mathbb{H}_{d}}$ defined in \eqref{eq:Theta-H_d} is {\em not} invariant under the action $\Phi$ and thus the simplified formula $\ip{ \mathbf{J}(\,\cdot\,) }{ \xi } = \Theta_{\mathbb{H}_{d}}\parentheses{ \xi_{\mathbb{H}_{d}}(\,\cdot\,) }$ (see, e.g., \citet[Theorem~4.2.10 on p.~282]{AbMa1978}) for the momentum map $\mathbf{J}$ is {\em not} valid here.
  Hence we will use the formula \eqref{eq:J-HamSys} to find the momentum map $\mathbf{J}$.
\end{remark}

Now our main result is the following:
\begin{theorem}
  \label{thm:Kostant-Siegel}
  Let $\mathbf{J}\colon \mathbb{H}_{d} \to \sp(\R^{2d})^{*} \cong \sym(2d,\R)$ be the momentum map~\eqref{def:J} corresponding to the $\Sp(\R^{2d})$ action $\Phi$ (see \eqref{eq:Sp_action}) on the Siegel upper half space $\mathbb{H}_{d}$.
  \begin{enumerate}[(i)]
  \item The image $\mathbf{J}(\mathcal{C})$ of $\mathcal{C} = \mathcal{A} + \rmi\mathcal{B} \in \mathbb{H}_{d}$ is the covariance matrix $\sigma(\mathcal{C})$ in the Gaussian state Wigner function~\eqref{eq:Wigner-psi_0}, i.e., 
    \begin{equation}
      \label{eq:J}
      \mathbf{J}(\mathcal{C}) \defeq 
      \begin{bmatrix}
        \mathcal{B}^{-1} & \mathcal{B}^{-1} \mathcal{A} \\
        \mathcal{A} \mathcal{B}^{-1} & \mathcal{A} \mathcal{B}^{-1} \mathcal{A} + \mathcal{B}
      \end{bmatrix} = \sigma(\mathcal{C}).
    \end{equation}
  \item $\mathbf{J}$ is an equivariant momentum map, i.e., for any $S \in \Sp(\R^{2d})$,
    \begin{equation}
      \label{eq:equivariance}
      \mathbf{J} \circ \Phi_{S} = \Ad^{*}_{S^{-1}} \mathbf{J}.
    \end{equation}
  \item $\mathbf{J}$ is a Poisson map with respect to the Poisson bracket~\eqref{eq:PB-H_d} and the $(+)$-Lie--Poisson bracket~\eqref{eq:LPB-sym}, i.e.,
    \begin{equation}
      \label{eq:J-Poisson}
      \PB{F}{G}_{\sym(2d,\R)}^{+} \circ \mathbf{J} = \PB{ F \circ \mathbf{J} }{ G \circ \mathbf{J} }_{\mathbb{H}_{d}}.
    \end{equation}
  \item The pull-back by $\mathbf{J}$ of the KKS symplectic form $\Omega_{\mathcal{O}}^{+}$ (see \eqref{eq:KKS-sym}) on a coadjoint orbit $\mathcal{O} \subset \sp(\R^{2d})^{*} \cong \sym(2d,\R)$ is the symplectic form $\Omega_{\mathbb{H}_{d}}$ (see \eqref{eq:Omega-H_d}) on the Siegel upper half space, i.e., $\mathbf{J}^{*}\Omega_{\mathcal{O}}^{+} = \Omega_{\mathbb{H}_{d}}$.
  \end{enumerate}
\end{theorem}

\begin{proof}
  Let us first find an expression for the momentum map $\mathbf{J}$.
  Set $\xi = \begin{tbmatrix}
    \xi_{11} & \xi_{12} \\
    \xi_{12}^{T} & \xi_{22}
  \end{tbmatrix} \in \sym_{d}(\R)$; then, using the expression~\eqref{eq:Sp_action} for $\Phi$ and writing $\mathcal{C} = \mathcal{A} + \rmi\mathcal{B}$, we have
  \begin{align*}
    \xi_{\mathbb{H}_{d}}(\mathcal{C})
    &= \parentheses{
      \tilde{\xi}_{21} + \tilde{\xi}_{22} \mathcal{A} - \mathcal{A} (\tilde{\xi}_{11} + \tilde{\xi}_{12} \mathcal{A}) + \mathcal{B} \tilde{\xi}_{12} \mathcal{B}
      }_{jk}\, \pd{}{\mathcal{A}_{jk}} \\
    &\quad - \parentheses{
      \mathcal{B} \tilde{\xi}_{11} + \tilde{\xi}_{11}^{T} \mathcal{B} + \mathcal{A} \tilde{\xi}_{12} \mathcal{B} + \mathcal{B} \tilde{\xi}_{12} \mathcal{A}
      }_{jk}\, \pd{}{\mathcal{B}_{jk}} \\
    &= \parentheses{
      \xi_{11} + \xi_{12} \mathcal{A} + \mathcal{A} \xi_{12}^{T} + \mathcal{A} \xi_{22} \mathcal{A} - \mathcal{B} \xi_{22} \mathcal{B}
      }_{jk}\, \pd{}{\mathcal{A}_{jk}} \\
    &\quad + \parentheses{
      \mathcal{B} \xi_{12}^{T} + \xi_{12} \mathcal{B} + \mathcal{B} \xi_{22} \mathcal{A} + \mathcal{A} \xi_{22} \mathcal{B}
      }_{jk}\, \pd{}{\mathcal{B}_{jk}}
  \end{align*}
  where we used the identities~\eqref{eq:xi-tildexi}.
  Then, using the definition~\eqref{eq:Omega-H_d} of the symplectic form $\Omega_{\mathbb{H}_{d}}$, we have
  \begin{align*}
    \ins{\xi_{\mathbb{H}_{d}}}\Omega_{\mathbb{H}_{d}}
    &= \brackets{ \mathcal{B}^{-1} (\mathcal{B} \xi_{12}^{T} + \xi_{12} \mathcal{B} + \mathcal{B} \xi_{22} \mathcal{A} + \mathcal{A} \xi_{22} \mathcal{B}) \mathcal{B}^{-1} }_{jk} \d\mathcal{A}_{jk} \\
    &\quad - \brackets{ 
      \mathcal{B}^{-1} (\xi_{11} + \xi_{12} \mathcal{A} + \mathcal{A} \xi_{12}^{T} + \mathcal{A} \xi_{22} \mathcal{A} - \mathcal{B} \xi_{22} \mathcal{B}) \mathcal{B}^{-1}
      }_{lm} \d\mathcal{B}_{lm} \\
    &= \tr\brackets{
      2( \xi_{12}^{T} + \xi_{22} \mathcal{A} ) \mathcal{B}^{-1} \d\mathcal{A}
      + (\xi_{11} + 2 \xi_{12} \mathcal{A} + \mathcal{A} \xi_{22} \mathcal{A}) \d\mathcal{B}^{-1} + \xi_{22} \d\mathcal{B}
      } \\
    &= \d\tr\brackets{
      \mathcal{B}^{-1} \xi_{11} + 2 \mathcal{A} \mathcal{B}^{-1} \xi_{12}
      + ( \mathcal{A} \mathcal{B}^{-1} \mathcal{A} + \mathcal{B}) \xi_{22}
      }.
  \end{align*}
  But then one notices that
  \begin{align*}
    \tr\brackets{
      \mathcal{B}^{-1} \xi_{11} + 2 \mathcal{A} \mathcal{B}^{-1} \xi_{12}
      + ( \mathcal{A} \mathcal{B}^{-1} \mathcal{A} + \mathcal{B} ) \xi_{22}
    }
    &= \tr\parentheses{
      \begin{bmatrix}
        \mathcal{B}^{-1} & \mathcal{B}^{-1} \mathcal{A} \\
        \mathcal{A} \mathcal{B}^{-1} & \mathcal{A} \mathcal{B}^{-1} \mathcal{A} + \mathcal{B}
      \end{bmatrix}
      \begin{bmatrix}
        \xi_{11} & \xi_{12} \\
        \xi_{12}^{T} & \xi_{22}
      \end{bmatrix}
     } \\
    &= \tr\parentheses{ \mathbf{J}(\mathcal{C}) \xi } \\
    &= \ip{ \mathbf{J}(\mathcal{C}) }{ \xi }
  \end{align*}
  with $\mathbf{J}(\mathcal{C})$ defined as in \eqref{eq:J}.
  Hence we have $\ins{\xi_{\mathbb{H}_{d}}}\Omega_{\mathbb{H}_{d}} = \d\ip{\mathbf{J}(\,\cdot\,)}{\xi}$ and thus \eqref{eq:J} gives the momentum map $\mathbf{J}$.
  We see that $\mathbf{J}(\mathcal{C})$ is nothing but the covariance matrix $\sigma(\mathcal{C})$ in \eqref{eq:sigma}.

  Let us next show the equivariance~\eqref{eq:equivariance}.
  First define, as in \eqref{eq:S-mathcalAB},
  \begin{equation*}
    \Xi\colon \mathbb{H}_{d} \to \Sp(\R^{2d});
    \qquad
    \mathcal{C} = \mathcal{A} + \rmi\mathcal{B} \mapsto
    \begin{bmatrix}
      \mathcal{B}^{-1/2} & 0 \\
      \mathcal{A}\mathcal{B}^{-1/2} & \mathcal{B}^{1/2}
    \end{bmatrix}.
  \end{equation*}
  As we have seen in \eqref{eq:S-mathcalAB2}, $\pi_{\U(d)} \circ \Xi(\mathcal{C}) = \mathcal{C}$ for any $\mathcal{C} \in \mathbb{H}_{d}$.
  Let us also define
  \begin{equation}
    \label{eq:Jhat}
    \hat{\mathbf{J}}\colon \Sp(\R^{2d}) \to \sym(2d,\R);
    \qquad
    Y \mapsto Y Y^{T}.
  \end{equation}
  Then $\mathbf{J} \circ \pi_{\U(d)} = \hat{\mathbf{J}}$, i.e., the diagram below commutes.
  \begin{equation*}
    \begin{tikzcd}[column sep=8ex, row sep=7ex]
      \Sp(\R^{2d}) \arrow{dr}{\hat{\mathbf{J}}} \arrow{d}[swap]{\pi_{\U(d)}} & 
      \\
      \mathbb{H}_{d} \arrow{r}[swap]{\mathbf{J}} & \sym(2d,\R)
    \end{tikzcd}
  \end{equation*}
  In fact, let $Y \in \Sp(\R^{2d})$ be arbitrary and set $\mathcal{C} = \mathcal{A} + \rmi\mathcal{B} = \pi_{\U(d)}(Y) \in \mathbb{H}_{d}$.
  Then $Y = \Xi(\mathcal{C})\, \mathcal{U}$ with some $\mathcal{U} \in \Sp(\R^{2d}) \cap \Orth(2d) \cong \U(d)$ and hence a simple calculation shows that
  \begin{equation*}
    \hat{\mathbf{J}}(Y) = \Xi(\mathcal{C})\, \Xi(\mathcal{C})^{T} = \mathbf{J}(\mathcal{C}) = \mathbf{J} \circ \pi_{\U(d)}(Y).
  \end{equation*}
  Now we see from the definition~\eqref{eq:Jhat} of $\hat{\mathbf{J}}$ that, for any $S, Y \in \Sp(\R^{2d})$, 
  \begin{equation*}
    \hat{\mathbf{J}} \circ L_{S}(Y) = S Y Y^{T} S^{T} = \Ad^{*}_{S^{-1}}\hat{\mathbf{J}}(Y),
  \end{equation*}
  where $\Ad^{*}$ on the right-hand side is defined in \eqref{eq:Adstar_on_Sym}.
  Let $\mathcal{C}$ be an arbitrary element in $\mathbb{H}_{d}$; then there exists $Y \in \Sp(\R^{2d})$ such that $\pi_{\U(d)}(Y) = \mathcal{C}$ since $\mathbb{H}_{d}$ is identified with the homogeneous space $\Sp(\R^{2d})/\U(d)$ as explained in Section~\ref{ssec:Siegel}.
  Consequently, upon using \eqref{eq:Phi-pi} and the above identities, we see that
  \begin{align*}
    \mathbf{J} \circ \Phi_{S}(\mathcal{C})
    &= \mathbf{J} \circ \Phi_{S} \circ \pi_{\U(d)}(Y) \\
    &= \mathbf{J} \circ \pi_{\U(d)} \circ L_{S}(Y) \\
    &= \hat{\mathbf{J}} \circ L_{S}(Y) \\
    &= \Ad^{*}_{S^{-1}} \hat{\mathbf{J}}(Y) \\
    &= \Ad^{*}_{S^{-1}} \mathbf{J} \circ \pi_{\U(d)}(Y) \\
    &= \Ad^{*}_{S^{-1}} \mathbf{J}(\mathcal{C}).
  \end{align*}
  Hence the equivariance~\eqref{eq:equivariance} follows.
  The diagram below summarizes the proof.
  \begin{equation*}
    \begin{tikzcd}[row sep=5ex, column sep=6ex]
      \Sp(\R^{2d}) \arrow{rrr}{L_{S}} \arrow{dd}[swap]{\pi_{\U(d)}} \arrow{rd}{\hat{\mathbf{J}}} & & & \Sp(\R^{2d}) \arrow{ld}[swap]{\hat{\mathbf{J}}} \arrow{dd}{\pi_{\U(d)}} \\
      & \sym(2d,\R) \arrow{r}{\Ad^{*}_{S^{-1}}} & \sym(2d,\R) & \\
      \mathbb{H}_{d} \arrow{rrr}[swap]{\Phi_{S}} \arrow{ru}{\mathbf{J}} & & & \mathbb{H}_{d} \arrow{lu}[swap]{\mathbf{J}}
    \end{tikzcd}
  \end{equation*}

  The equivariance of $\mathbf{J}$ implies that $\mathbf{J}$ is Poisson in the sense described in the statement; see, e.g., \citet[Theorem~12.4.1 on p.~403]{MaRa1999}.
 
  That $\mathbf{J}^{*}\Omega_{\mathcal{O}}^{+} = \Omega_{\mathbb{H}_{d}}$ follows from Kostant's coadjoint orbit covering theorem~\cite{Ko1966} (see also \citet[Theorem~14.4.5 on p.~465]{MaRa1999}) because $\Phi$ is a left transitive action as we discussed in Section~\ref{ssec:Siegel}.
\end{proof}

\section{Collective Dynamics of Covariance Matrix}
\label{sec:collective_dynamics}
Theorem~\ref{thm:Kostant-Siegel} suggests that the $\mathbb{H}_{d}$ portion of the Gaussian wave packet dynamics~\eqref{eq:HamiltonianSystem-GWP} on $\mathcal{P} \defeq \R^{2d} \times \mathbb{H}_{d}$ becomes a Lie--Poisson dynamics on $\sym(2d,\R)$ via the momentum map $\mathbf{J}$ defined in \eqref{eq:J}.
In other words, the dynamics of the corresponding covariance matrix $\Sigma = \sigma(\mathcal{C})$ is an example of the so-called {\em collective dynamics} (see, e.g., \citet{GuSt1980,GuSt1990}).

In this section, we assume that the potential $V$ is quadratic for simplicity; a more general case will be discussed in the next section.
Notice that, when $V$ is quadratic, \eqref{eq:OhLe} decouples into \eqref{eq:OhLe-qp}---which becomes a classical system---and \eqref{eq:OhLe-AB}, i.e., the dynamics decouples into those on $\R^{2d}$ and $\mathbb{H}_{d}$.
Hence we focus on the dynamics on $\mathbb{H}_{d}$ and $\sym(2d,\R)$ for now and incorporate the $\R^{2d}$ portion in the next section.
Specifically, we may define the Hamiltonian $H_{\mathbb{H}_{d}}\colon \mathbb{H}_{d} \to \R$ by
\begin{equation*}
  H_{\mathbb{H}_{d}} \defeq -\tr\brackets{ \mathcal{B}^{-1}\parentheses{ \frac{\mathcal{A}^{2} + \mathcal{B}^{2}}{m} + D^{2}V } }.
\end{equation*}
Then the Hamiltonian system $\ins{X_{\mathbb{H}_{d}}} \Omega_{\mathbb{H}_{d}} = \d{H_{\mathbb{H}_{d}}}$ gives \eqref{eq:OhLe-AB}.
As we shall see below, this Hamiltonian $H_{\mathbb{H}_{d}}$ turns out to be a collective Hamiltonian with the momentum map $\mathbf{J}$, and so the corresponding dynamics becomes Lie--Poisson via the momentum map $\mathbf{J}$:
\begin{corollary}
  \label{cor:collective_dynamics}
  Let $H_{\rm cl}\colon \R^{2d} \to \R$ be the classical Hamiltonian
  \begin{equation}
    \label{eq:H_cl}
    H_{\rm cl}(q,p) = \frac{1}{2m}p^{2} + V(q),
  \end{equation}
  and suppose that $V$ is quadratic.
  Also define $h_{\sym}\colon \sym(2d,\R) \cong \sp(\R^{2d})^{*} \to \R$ by
  \begin{equation*}
    h_{\sym}(\Sigma) \defeq -\ip{ \Sigma }{ D^{2}H_{\rm cl} } = -\tr( \Sigma D^{2}H_{\rm cl} ),
  \end{equation*}
  where $\Sigma = \begin{tbmatrix}
    \Sigma_{11} & \Sigma_{12} \\
    \Sigma_{12}^{T} & \Sigma_{22}
  \end{tbmatrix} \in \sym(2d,\R) \cong \sp(\R^{2d})^{*}$.
  Then $H_{\mathbb{H}_{d}}$ can be written as a collective Hamiltonian using $h_{\sym}$ and the momentum map $\mathbf{J}$, i.e.,
  \begin{equation*}
    H_{\mathbb{H}_{d}} = h_{\sym} \circ \mathbf{J}.
  \end{equation*}
  As a result, the vector field $X_{\mathcal{O}}$ defined by $X_{\mathcal{O}} \circ \mathbf{J} = T\mathbf{J} \circ X_{\mathbb{H}_{d}}$ on the coadjoint orbit $\mathcal{O}_{\Sigma} \subset \sym(2d,\R) \cong \sp(\R^{2d})^{*}$ through $\Sigma = \mathbf{J}(\mathcal{C})$ is the Lie--Poisson dynamics defined by the above Hamiltonian $h_{\sym}$, i.e.,
  \begin{equation}
    \label{eq:Lie-Poisson}
    \dot{\Sigma} = -\ad^{*}_{\delta h_{\sym}/\delta \Sigma} \Sigma = \mathbb{J} D^{2}H_{\rm cl}\, \Sigma - \Sigma D^{2}H_{\rm cl}\, \mathbb{J}.
  \end{equation}
\end{corollary}

\begin{proof}
  It is easy to see that $D^{2}H_{\rm cl} =
  \begin{tbmatrix}
    D^{2}V & 0 \\
    0 & I_{d}/m
  \end{tbmatrix}$ and thus
  \begin{equation*}
    h_{\sym}(\Sigma) = -\tr\parentheses{ \frac{\Sigma_{22}}{m} + \Sigma_{11}D^{2}V }.
  \end{equation*}
  Then a simple calculation shows that $H_{\mathbb{H}_{d}}(\mathcal{C}) = h_{\sym} \circ \mathbf{J}(\mathcal{C})$ for any $\mathcal{C} \in \mathbb{H}_{d}$.

  Now, by the Collective Hamiltonian Theorem (see, e.g., \citet[Theorem~12.4.2]{MaRa1999}), the vector field $X_{\mathbb{H}_{d}}$ satisfies, for any $\mathcal{C} \in \mathbb{H}_{d}$,
  \begin{equation*}
    X_{\mathbb{H}_{d}}(\mathcal{C}) = X_{h_{\sym} \circ \mathbf{J}}(\mathcal{C}) = \parentheses{ \Fd{h_{\sym}}{\Sigma} }_{\mathbb{H}_{d}}(\mathcal{C}),
  \end{equation*}
  where $\parentheses{ \delta h_{\sym}/\delta\Sigma }_{\mathbb{H}_{d}}$ stands for the infinitesimal generator of $\delta h_{\sym}/\delta\Sigma \in \sym(2d,\R) \cong \sp(\R^{2d})$.
  Then we have
  \begin{equation*}
    T_{z}\mathbf{J} \cdot X_{\mathbb{H}_{d}}(\mathcal{C}) = T_{z}\mathbf{J} \cdot \parentheses{ \Fd{h_{\sym}}{\Sigma} }_{\mathbb{H}_{d}}(\mathcal{C}).
  \end{equation*}
  However, by the equivariance $\mathbf{J} \circ \Phi_{S} = \Ad_{S^{-1}}^{*} \mathbf{J}$ of the momentum map for any $S \in \Sp(\R^{2d})$, we obtain, for any $\xi \in \sp(\R^{2d})$,
  \begin{equation*}
    T_{z}\mathbf{J} \circ \xi_{\mathbb{H}_{d}}(\mathcal{C}) = -\ad_{\xi}^{*} \mathbf{J}(\mathcal{C}),
  \end{equation*}
  and thus
  \begin{equation*}
    T_{z}\mathbf{J} \cdot X_{\mathbb{H}_{d}}(\mathcal{C}) = -\ad_{\delta h_{\sym}/\delta\Sigma}^{*} \mathbf{J}(\mathcal{C}) = X_{\mathcal{O}} \circ \mathbf{J}(\mathcal{C})
  \end{equation*}
  with 
  \begin{equation*}
    X_{\mathcal{O}}(\Sigma) \defeq -\ad^{*}_{\delta h_{\sym}/\delta \Sigma} \Sigma.
  \end{equation*}
  But then this is nothing but the Hamiltonian vector field defined by $\ins{X_{\mathcal{O}}} \Omega_{\mathcal{O}}^{+} = \d{h_{\sym}}$ on the coadjoint orbit $\mathcal{O}_{\Sigma}$, where $\Omega_{\mathcal{O}}^{+}$ is the KKS symplectic form~\eqref{eq:KKS-sym}.
  Then the formula~\eqref{eq:ad_star-Sym} for $\ad^{*}$ with $\delta h_{\sym}/\delta \Sigma = -D^{2}H_{\rm cl}$ yields \eqref{eq:Lie-Poisson}.
\end{proof}

\begin{remark}
  The Lie--Poisson equation~\eqref{eq:Lie-Poisson} is compatible with the dynamics on $\Sp(\R^{2d})$ due to \citet{Ha1980, Ha1981, Hagedorn1985, Ha1998} (see also \citet[Section~7]{Li1986} and \citet[Section~V.1]{Lu2008}).
  Hagedorn parametrizes an element $S \in \Sp(\R^{2d})$ such that $\pi_{\U(d)}(S) = \mathcal{C} = \mathcal{A} + \rmi\mathcal{B}$ as $S =
  \begin{tbmatrix}
    \Re Q & \Im Q \smallskip\\
    \Re P & \Im P
  \end{tbmatrix}$ with $Q, P \in \Mat_{d}(\C),$ $Q^{T}P - P^{T}Q = 0$, and $Q^{*}P - P^{*}Q = 2\rmi I_{d}$; these conditions are equivalent to $S \in \Sp(\R^{2d})$; see \eqref{def:Sp2d-block}.
  Then the equations~\eqref{eq:OhLe-AB} for $\mathcal{A}$ and $\mathcal{B}$ are replaced by
  \begin{equation*}
    \dot{Q} = \frac{P}{m},
    \qquad
    \dot{P} =  -D^{2}V\,Q,
  \end{equation*}
  which are equivalent to
  \begin{equation}
    \label{eq:S_equation}
    \dot{S} = \mathbb{J} D^{2}H_{\rm cl}\, S.
  \end{equation}
  Since $\pi_{\U(d)}(S) = \mathcal{C}$, we have $\Sigma = \mathbf{J}(\mathcal{C}) = \hat{\mathbf{J}}(S) = S S^{T}$, and then \eqref{eq:S_equation} gives \eqref{eq:Lie-Poisson}.
\end{remark}

\section{Dynamics of Gaussian State Wigner Function}
\label{sec:DynamicsOfSCSWigner}
\subsection{Dynamics under Non-quadratic Potentials}
If the potential $V$ is not quadratic, the equations for $z = (q,p)$ and the covariance matrix $\Sigma$ must be coupled as in \eqref{varianceformulation}.
Hence the simple collectivization presented in the previous section does not provide the bridge between the Gaussian wave packet dynamics~\eqref{eq:OhLe} and the Gaussian Wigner dynamics \eqref{varianceformulation}.
However, fortunately, it turns out that a simple modification of the approach from the previous section provides the desired bridge between them.

As mentioned in Section~\ref{ssec:Symplectic_GWPD}, the dynamics of the Gaussian wave packet~\eqref{eq:chi} is reduced to the Hamiltonian system~\eqref{eq:HamiltonianSystem-GWP} with symplectic form~\eqref{eq:Omega} and Hamiltonian~\eqref{eq:H} defined on $\mathcal{P} = \R^{2d} \times \mathbb{H}_{d}$.
One may write the symplectic form~\eqref{eq:Omega} on $\mathcal{P} = \R^{2d} \times \mathbb{H}_{d}$ as
\begin{equation*}
  \Omega = \pr_{\R^{2d}}^{*} \Omega_{\R^{2d}} - \frac{\hbar}{4} \pr_{\mathbb{H}_{d}}^{*} \Omega_{\mathbb{H}_{d}}
\end{equation*}
with the natural projections $\pr_{\R^{2d}}\colon \mathcal{P} \to \R^{2d}$ and $\pr_{\mathbb{H}_{d}}\colon \mathcal{P} \to \mathbb{H}_{d}$, whereas the corresponding Poisson bracket on $\mathcal{P}$ is given by \eqref{eq:PB-GWP}.
It is straightforward to adapt Theorem~\ref{thm:Kostant-Siegel} to this setting to relate the Gaussian wave packet dynamics~\eqref{eq:OhLe} with the dynamics~\eqref{varianceformulation} of the Gaussian Wigner function.
\begin{proposition}
  \label{prop:Kostant-GWP}
  Let $\mathsf{G} \defeq \R^{2d} \times \Sp(\R^{2d})$ and $\mathcal{P} = \R^{2d} \times \mathbb{H}_{d}$, and let $\Psi\colon \mathsf{G} \times \mathcal{P} \to \mathcal{P}$ be the $\mathsf{G}$-action on $\mathcal{P}$ defined by
  \begin{equation*}
    \Psi_{(\delta z, S)}(z, \mathcal{C}) \defeq (z + \mathbb{J}\delta z, \Phi_{S}(\mathcal{C})).
  \end{equation*}
  Then:
  \begin{enumerate}[(i)]
  \item The corresponding momentum map $\mathbf{M}\colon \mathcal{P} \to \R^{2d} \times \sp(\R^{2d})^{*} \cong \R^{2d} \times \sym(2d,\R)$ is given by
    \begin{equation*}
      \mathbf{M}(z, \mathcal{C}) = (z, \sigma(\mathcal{C}))
    \end{equation*}
    and is equivariant, where $\sigma$ is defined in \eqref{eq:sigma}.
    \smallskip
  \item $\mathbf{M}$ is Poisson with respect to the Poisson brackets~\eqref{eq:PB-GWP} and \eqref{eq:PB-GaussianWigner}.
    \smallskip
  \item Let $\mathcal{O} \subset \sp(\R^{2d})^{*} \cong \sym(2d,\R)$ be a coadjoint orbit and define a symplectic form $\omega$ on $\R^{2d} \times \mathcal{O}$ by
    \begin{equation*}
      \omega \defeq \pi_{\R^{2d}}^{*}\Omega_{\R^{2d}} - \frac{\hbar}{4}\pi_{\mathcal{O}}^{*}\Omega_{\mathcal{O}}^{+},
    \end{equation*}
    where $\pi_{\R^{2d}}\colon \R^{2d} \times \mathcal{O} \to \R^{2d}$ and $\pi_{\mathcal{O}}\colon \R^{2d} \times \mathcal{O} \to \mathcal{O}$ are natural projections.
    Then the pull-back by $\mathbf{M}$ of the symplectic form $\omega$ is the symplectic form $\Omega$ in \eqref{eq:Omega} on $\mathcal{P}$, i.e., $\mathbf{M}^{*}\omega = \Omega$.
    \smallskip
  \item Let $\mathcal{O}_{\Sigma}$ be the coadjoint orbit $\mathcal{O}_{\Sigma} \subset \sym(2d,\R) \cong \sp(\R^{2d})^{*}$ through $\Sigma = \sigma(\mathcal{C})$.
    The vector field $X_{h}$ on $\R^{2d} \times \mathcal{O}_{\Sigma}$ defined by $X_{h} \circ \mathbf{M} = T\mathbf{M} \circ X_{H}$, where $X_{H}$ is given in \eqref{eq:HamiltonianSystem-GWP}, is the Hamiltonian vector field with respect to the above symplectic form $\omega$ and the Hamiltonian
    \begin{equation}
      \label{eq:h}
      h(z, \Sigma) \defeq H_{\rm cl}(z) - \frac{\hbar}{4}h_{\sym}(\Sigma)
      = \frac{p^{2}}{2m} + V(q) + \frac{\hbar}{4} \tr\parentheses{ \frac{\Sigma_{22}}{m} + \Sigma_{11}D^{2}V(q) },
    \end{equation}
    i.e., $\ins{X_{h}} \omega = \d{h}$, and is given by \eqref{varianceformulation}, or more concretely
    \begin{equation}
      \label{eq:BoTr}
      \begin{array}{c}
        \DS \dot{q} = \frac{p}{m},
        \qquad
        \DS \dot{p} = -\pd{}{q}\brackets{ V(q) + \frac{\hbar}{4} \tr\parentheses{ \Sigma_{11}D^{2}V(q) } },
        \medskip\\
        \DS \dot{\Sigma} = -\ad^{*}_{\delta h/\delta \Sigma} \Sigma = \mathbb{J} D^{2}H_{\rm cl}(z)\, \Sigma - \Sigma D^{2}H_{\rm cl}(z)\, \mathbb{J}.
      \end{array}
    \end{equation}
  \end{enumerate}
\end{proposition}

\begin{proof}
  The expression of the momentum map $\mathbf{M}$ follows from a straightforward calculation and the equivariance of $\mathbf{J}$ from Theorem~\ref{thm:Kostant-Siegel}.
  That $\mathbf{M}$ is Poisson is clear from the fact that $\mathbf{M} = \id_{\R^{2d}} \times \mathbf{J}$ as well as the definitions of the Poisson brackets and the fact that $\mathbf{J}$ is Poisson in the sense of \eqref{eq:J-Poisson}.
  Since the action $\Psi$ is clearly transitive and $\mathcal{P}$ is a symplectic manifold with the symplectic form $\Omega$ in \eqref{eq:Omega}, it follows that $\mathbf{M}^{*}\omega = \Omega$ again from Kostant's coadjoint orbit covering theorem.
  The last statement follows easily from the Collective Hamiltonian Theorem (see, e.g., \citet[Theorem~12.4.2]{MaRa1999}) following a similar argument as in Corollary~\ref{cor:collective_dynamics}.
\end{proof}

\begin{remark}
  Assuming some regularity and decay properties of the potential $V$, one can show that the Hamiltonian~\eqref{eq:h} is in fact an $O(\hbar^{2})$ approximation to the expectation value of the classical Hamiltonian~\eqref{eq:H_cl} with respect to the Gaussian Wigner function~\eqref{eq:W_0} (where $\Sigma$ is assumed to be positive-definite), i.e.,
  \begin{equation*}
    \exval{H_{\rm cl}}_{0}(z, \Sigma) = \int_{\R^{2d}} H_{\rm cl}(\zeta)\, \mathcal{W}_{0}(\zeta)\,\mathrm{d}\zeta
    = h(z, \Sigma) + O(\hbar^{2}),
  \end{equation*}
  just as \eqref{eq:H} is an $O(\hbar^{2})$ approximation to the expectation value $\bigl\langle \psi_{0}, \hat{H}\psi_{0} \bigr\rangle$.
\end{remark}

\subsection{Numerical Results---The Effect of the Correction Term}
As mentioned earlier, both \eqref{eq:OhLe} and \eqref{eq:BoTr} differ from those time evolution equations that appeared in the previous works \cite{He1975a,He1976b,Li1986,Ha1980,Ha1981,Hagedorn1985,Ha1998,CoRo2012} by the $O(\hbar)$ correction term to the potential; see the time evolution equation for the momentum $p$.
More specifically, the classical Hamiltonian system for $(q,p)$ with the potential $V(q)$ is replaced by that with the corrected potential
\begin{equation}
  \label{eq:V_hbar}
  V_{\hbar}(q,\mathcal{B}) \defeq V(q) + \frac{\hbar}{4} \tr\parentheses{ \mathcal{B}^{-1} D^{2}V(q) }
  \quad\text{or}\quad
  V_{\hbar}(q,\Sigma) \defeq V(q) + \frac{\hbar}{4} \tr\parentheses{ \Sigma_{11} D^{2}V(q) },
\end{equation}
where we call both of them $V_{\hbar}$ with an abuse of notation.
Notice that, as a result, the equations for $(q,p)$ are coupled with the rest of the system through the $O(\hbar)$ correction term.
What is the effect of the correction term?
Here we limit ourselves to numerical experiments and set aside the proof of the asymptotic error in $\hbar$ as future work.
Our test case is a two-dimensional problem, i.e., $d = 2$, with the torsional potential
\begin{equation*}
  \displaywidth=\parshapelength\numexpr\prevgraf+2\relax
  V(q^{1}, q^{2}) = 2 - \cos q^{1} - \cos q^{2}.
\end{equation*}
This is the type of potential used to model torsional forces between molecules; see Fig.~\ref{fig:Torsions} and, e.g., \citet[Section~2.2.4]{Je2007}.
\begin{figure}
  \centering
  \includegraphics[width=.3\linewidth]{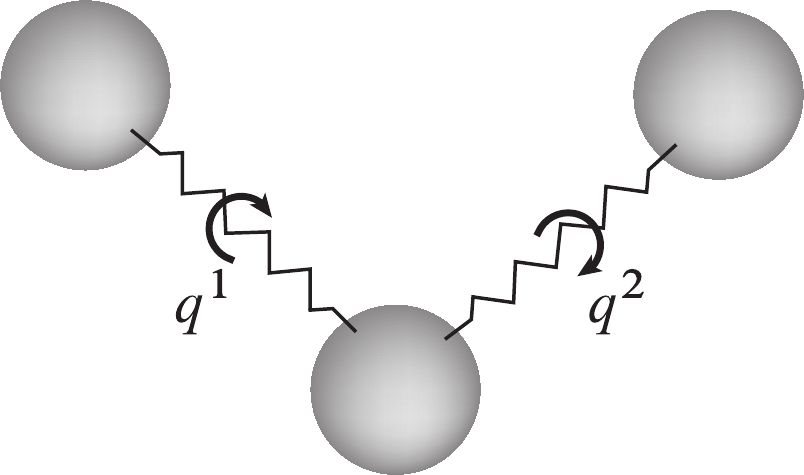}
  \caption{Torsions between molecules}
  \label{fig:Torsions}
\end{figure}
We compare the dynamics of the phase space variables $z = (q,p)$ to the dynamics of the expectation values $\exval{\hat{z}} \defeq (\exval{\hat{x}}, \exval{\hat{p}})$ of the standard position and momentum operators.
Directly solving the Schr\"odinger equation~\eqref{eq:Schroedinger} in the semiclassical regime $\hbar \ll 1$ numerically is a challenge due to its highly oscillatory solutions.
Therefore, as an effective alternative, we used Egorov's Theorem~\cite{Eg1969,CoRo2012,LaRo2010} or the Initial Value Representation (IVR) method~\cite{Mi1970,Mi1974b,WaSuMi1998,Mi2001} to compute the time evolution of the expectation values $\exval{\hat{z}}$; it is known that the Egorov/IVR method gives an $O(\hbar^{2})$ approximation to the exact evolution of the expectation values.

We set $m = 1$ and choose the initial condition
\begin{equation*}
  q(0) = (1,0),
  \qquad
  p(0) = (-1,1),
  \qquad
  \mathcal{A}(0) =
  \mathcal{B}(0) =
  \begin{bmatrix}
    1 & 0.5 \\
    0.5 & 1
  \end{bmatrix}.
\end{equation*}
For the Egorov/IVR method, the initial Wigner function is the Gaussian~\eqref{eq:Wigner-psi_0} corresponding to the initial condition.
Note that, by Proposition~\ref{prop:Kostant-GWP}, the Gaussian wave packet dynamics~\eqref{eq:OhLe} and the Gaussian Wigner dynamics~\eqref{eq:BoTr} give the same dynamics for $(q(t),p(t))$.
We solved \eqref{eq:OhLe} using the variational splitting integrator of \citet{FaLu2006} (see also \citet[Section~IV.4]{Lu2008}) and used the St\"ormer--Verlet method~\cite{Ve1967} to solve the classical Hamiltonian system; it is known that the variational splitting integrator converges to the St\"ormer--Verlet method as $\hbar \to 0$~\cite{FaLu2006}.
The time step is $0.01$ in all the cases and 10,000 particles are sampled from the initial Wigner function for the Egorov/IVR calculations.

\begin{figure}[htbp]
  \centering
  \subfigure[Time evolution of $q^{1}$ for $\eps = 0.1$]{
    \includegraphics[width=.45\linewidth]{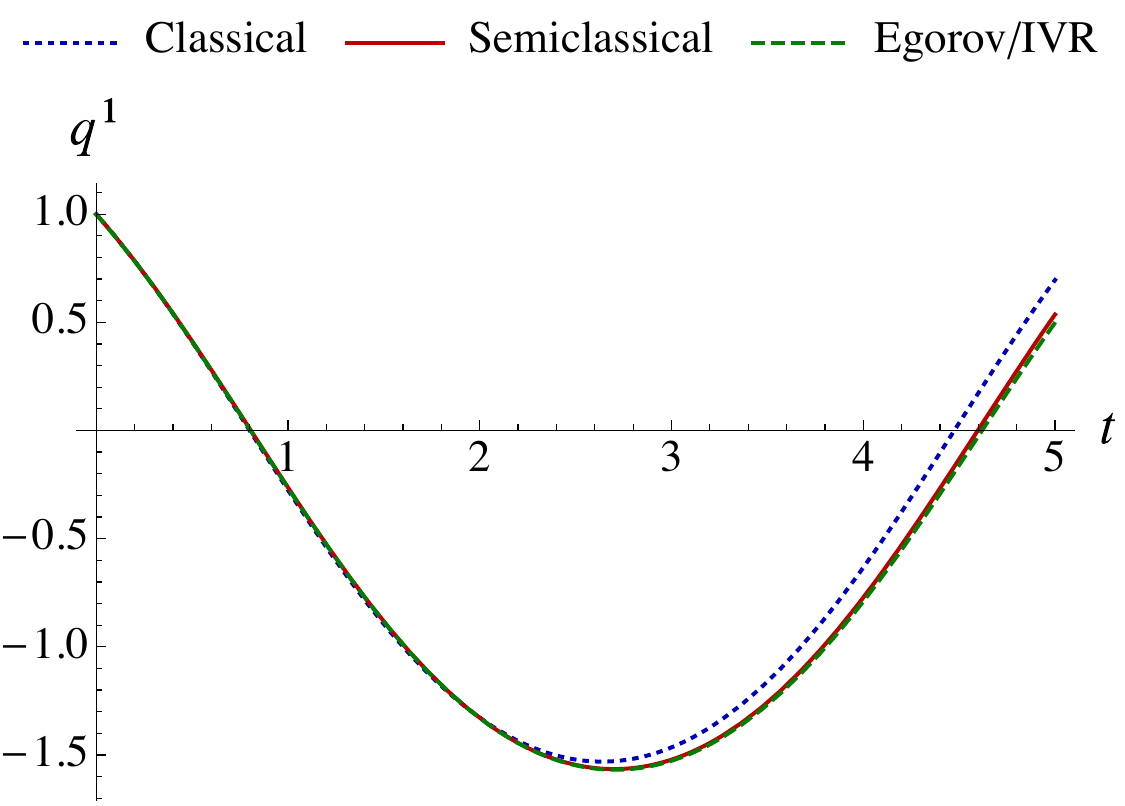}
    \label{fig:t-q1_01}
  }
  \qquad
  \subfigure[Time evolution of $p_{1}$ for $\eps = 0.1$]{
    \includegraphics[width=.45\linewidth]{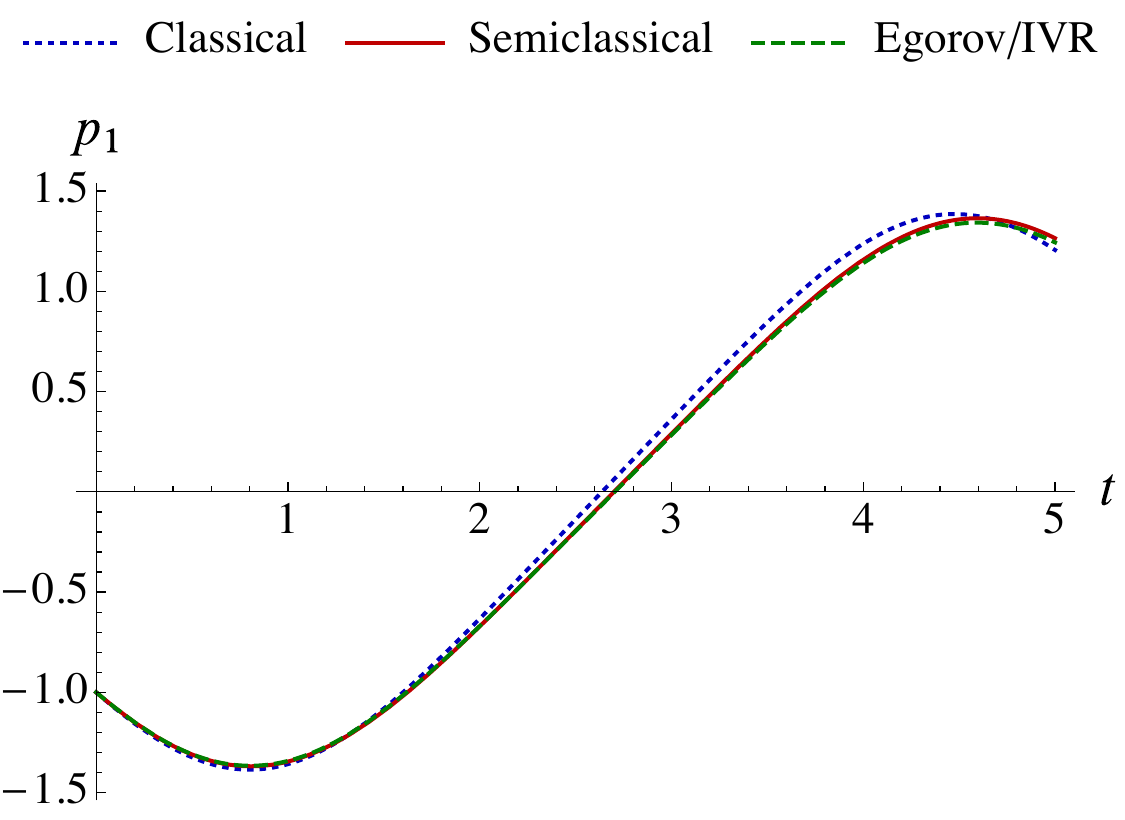}
    \label{fig:t-p1_01}
  }
  \medskip\\
  \subfigure[Convergence of error in observables $z = (q,p)$ as $\eps \to 0$]{
    \includegraphics[width=.465\linewidth]{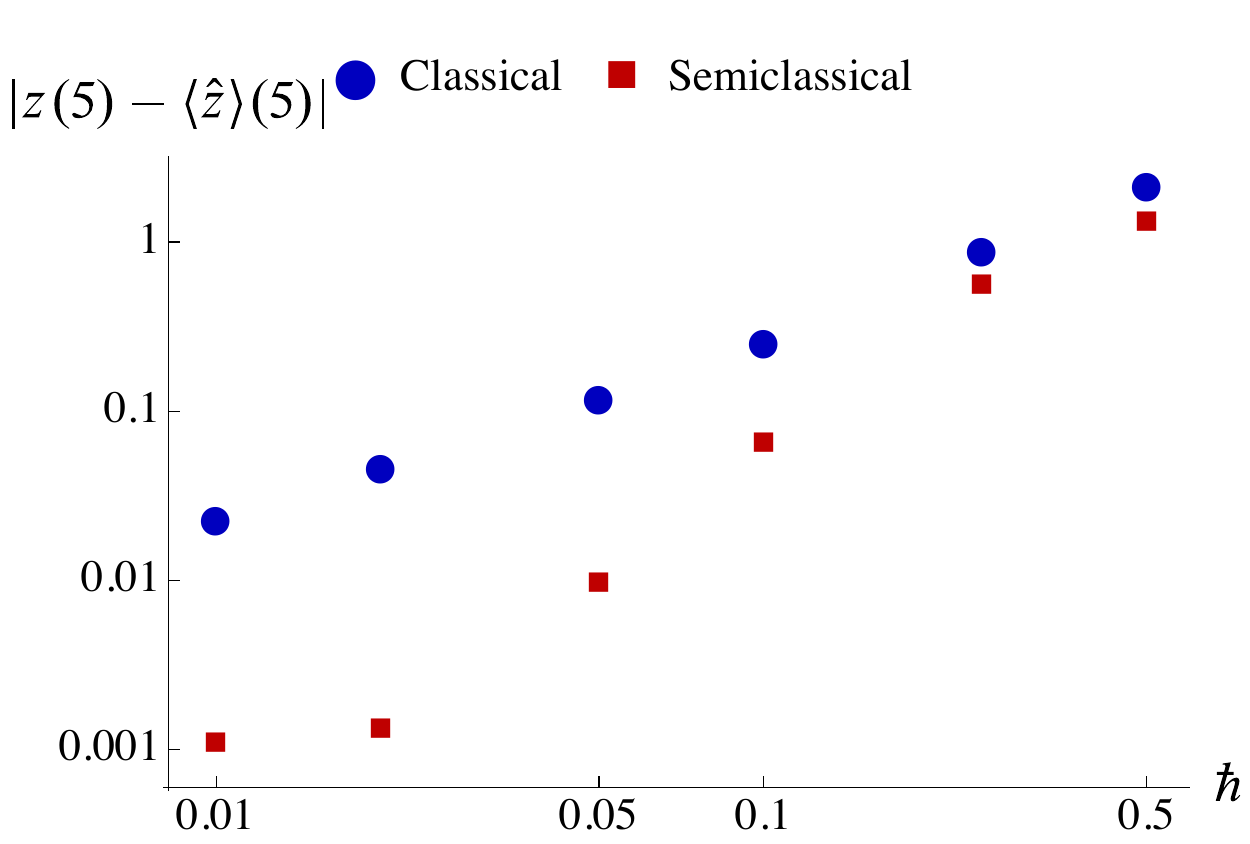}
    \label{fig:hbar-zerror}
  }
  \qquad
  \subfigure[$V(q_{\rm cl}(t))$, $V_{\eps}(q(t),\mathcal{B}(t))$, and $\exval{V}(t)$ for $\eps = 0.1$]{
    \includegraphics[width=.45\linewidth]{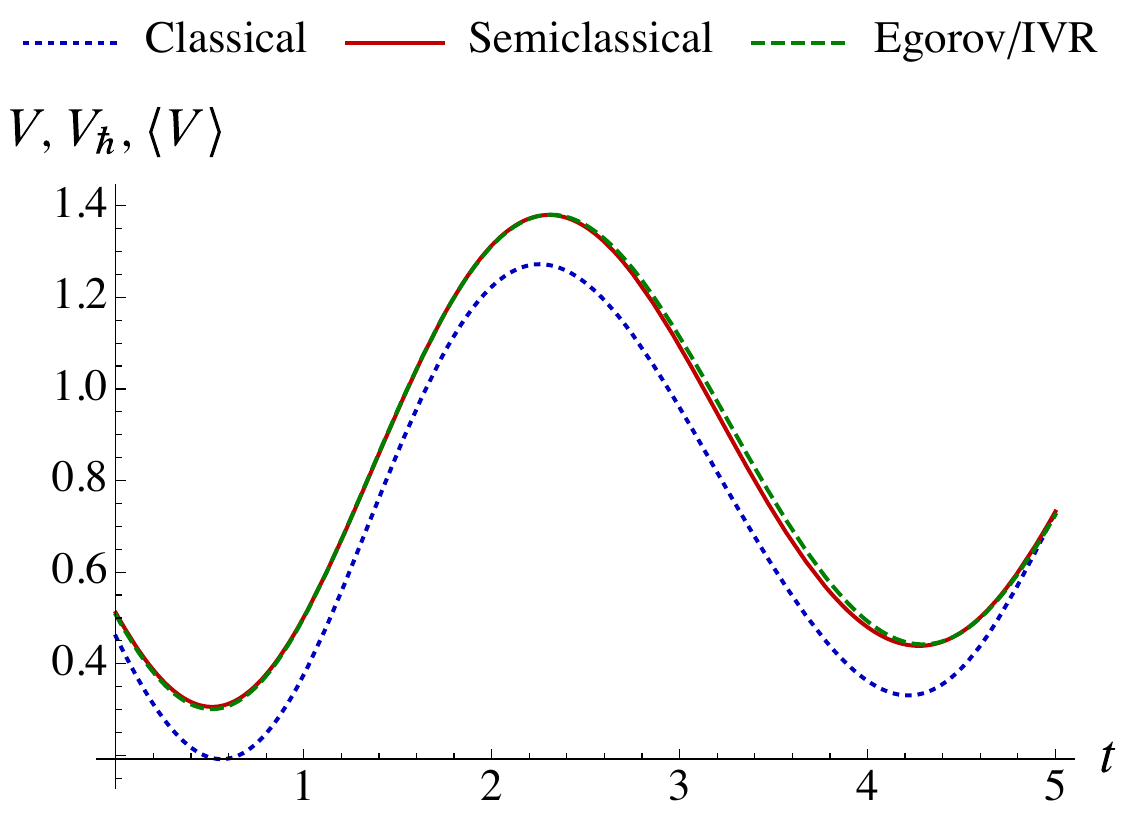}
    \label{fig:t-V_01}
  }
  \captionsetup{width=0.95\textwidth}
  \caption{\small
    (a)--(b)~Time evolution $(q^{1}(t),p_{1}(t))$ of classical, semiclassical, and Egorov/IVR solutions.
    The semiclassical solutions approximate the Egorov/IVR solutions much better than the classical solutions do.
    (c)~Errors in Euclidean norm in $T^{*}\R^{2} \cong \R^{4}$ of classical and semiclassical solution compared to expectation values obtained by Egorov/IVR algorithm at $t = 5$.
    The error $|z(5) - \exval{\hat{z}}(5)|$ is always smaller and converges faster as $\hbar \to 0$ with $z$ being the semiclassical solution than the classical one, indicating that the asymptotic error of $|z(5) - \exval{\hat{z}}(5)|$ as $\hbar \to 0$ is smaller with the semiclassical solution.
    (d)~Comparison of classical potential along classical solution, modified potential~\eqref{eq:V_hbar} along semiclassical solution, and expectation value of potential by Egorov/IVR.
    The modified potential with the semiclassical potential approximates the expectation value of the potential calculated using the Egorov/IVR method remarkably well, whereas the classical potential with the classical solutions deviates from it quite significantly.
  }
  \label{fig:CM-GWP-Egorov}
\end{figure}

The results (see Fig.~\ref{fig:CM-GWP-Egorov}) demonstrate that the semiclassical dynamics~\eqref{eq:OhLe} gives a better approximation to the expectation value dynamics of the Egorov algorithm compared to the classical solution $z_{\rm cl}(t) \defeq (q_{\rm cl}(t), p_{\rm cl}(t))$; recall that the classical solution has been commonly used for propagation of the phase space center of the coherent states~\cite{He1975a,He1976b,Li1986,Ha1980,Ha1981,Hagedorn1985,Ha1998,CoRo2012}.

In fact, as shown in Fig.~\ref{fig:hbar-zerror}, the error $|z(5) - \exval{\hat{z}}(5)|$ in Euclidean norm in $T^{*}\R^{2} \cong \R^{4}$ of semiclassical solution $z(t)$ of \eqref{eq:OhLe} converges faster than that of the classical solution $z_{\rm cl}(t)$ as $\hbar \to 0$.
The slowdown of the convergence of the error of the semiclassical solution around $\hbar = 0.01$ may be attributed to the lack of accuracy of the Egorov/IVR solution:
It involves a Monte-Carlo type numerical integration and thus the error is proportional to $1/\sqrt{N}$.
Since $N = 10,000$ here, a rough estimate of the error in Egorov/IVR solution due to the sampling is in the order of $0.01$.

Moreover, as shown in Fig.~\ref{fig:t-V_01}, the modified potential $V_{\hbar}$ from \eqref{eq:V_hbar} approximates the expectation value of the potential $\exval{V}$ with remarkable accuracy even for the relatively large $\hbar = 0.1$.
This result tends to justify our Hamiltonian formulation of the dynamics with the semiclassical Hamiltonians $H$ and $h$ and from \eqref{eq:H} and \eqref{eq:h} because $V_{\hbar}$ is nothing but the potential part of them.

\section*{Acknowledgments}
The authors are grateful to Fran\c{c}ois Gay-Balmaz, Darryl Holm, and Paul Skerritt for several discussions on this and related topics. C.T. acknowledges financial support by the Leverhulme Trust Research Project Grant 2014-112 and by the London Mathematical Society Grant No. 31320 (Applied Geometric Mechanics Network)

\appendix

\section{Moments of Wigner Function as a Momentum Map}
\label{appendix}
In this appendix, we find the expression~\eqref{eq:J_sstar} of the momentum map $\mathbf{J}_{\mathfrak{s}^{*}}$ corresponding to the action~\eqref{eq:Jac_action} of $\mathsf{Jac}(\R^{2d})$ on the space $\mathfrak{s}^{*}$ of Wigner functions, which is (formally) thought of as the dual of the space $\mathfrak{s}$ of observables.
Recall that we endowed the space $\mathfrak{s}^{*}$ of Wigner functions with the Lie--Poisson bracket \cite{BiMo1991}
\[
  \{F,K\}_{\mathfrak{s}^{*}}(\mathcal{W})=\int_{\R^{2d}}\!\mathcal{W}\left\{\!\!\left\{\frac{\delta F}{\delta \mathcal{W}},\frac{\delta K}{\delta \mathcal{W}}\right\}\!\!\right\}\,\mathrm{d}\zeta
\,,
\]
where $\{\!\{\cdot,\cdot\}\!\}$ denotes the Moyal bracket.
Upon considering a curve in $(S(s),{\sf z}(s),\theta(s))\in\mathsf{Jac}(\R^{2d})$ such that $(S(0),{\sf z}(0),\theta(0))=({\rm Id}, 0, 0)$ and $(S'(0),{\sf z}'(0),\theta'(0))=(\mathcal{S},\xi,\vartheta)$, we compute the infinitesimal generator of the action~\eqref{eq:Jac_action} as follows:
\[
\big((\mathcal{S},\xi,\vartheta)\cdot\mathcal{W}\big)(\zeta)=\frac{\mathrm{d}}{\mathrm{d}s}\bigg|_{s=0}\mathcal{W}(S(s)\zeta+{\sf z}(s))
=
(\mathcal{S}\zeta+\xi)\cdot\nabla\mathcal{W}(\zeta)
\,.
\]
Therefore, in order to show that \eqref{eq:J_sstar} is the momentum map $\mathbf{J}_{\mathfrak{s}^{*}}$ corresponding to the action~\eqref{eq:Jac_action}, we need to prove that 
\[
\left\{F,\,\frac{1}2\ip{\mathcal{S}}{\Bbb{J}^{T}\langle\zeta\otimes\zeta\rangle}_{\sp} +\xi\cdot\Bbb{J}^{T}\langle\zeta\rangle+\vartheta\langle1\rangle\right\}_{\mathfrak{s}^{*}}(\mathcal{W})
=\frac{\delta F}{\delta\mathcal{W}}\,(\mathcal{S}\zeta+\xi)\cdot\nabla\mathcal{W}(\zeta),
\]
where we used the expectation value notation~\eqref{eq:exval} as well as the inner product $\ip{\,\cdot\,}{\,\cdot\,}_{\sp}$ from \eqref{eq:ip-sp}.
This is verified by a direct calculation.
Indeed, we compute
\begin{align*}
&\int_{\R^{2d}} \mathcal{W}\left\{\!\!\left\{\frac{\delta F}{\delta \mathcal{W}},\frac{\delta}{\delta \mathcal{W}}\left(\frac{1}2\ip{\mathcal{S}}{\Bbb{J}^{T}\langle\zeta\otimes\zeta\rangle}_{\sp} +\xi\cdot\Bbb{J}^{T}\langle\zeta\rangle+\vartheta\langle1\rangle\right)\!\right\}\!\!\right\}\,\mathrm{d}\zeta
\\
=&\int_{\R^{2d}} \mathcal{W}\left\{\frac{\delta F}{\delta \mathcal{W}},\, \frac{1}2\tr( (\zeta\otimes\zeta)\,\Bbb{J}\mathcal{S}) +\xi\cdot\Bbb{J}^{T}\zeta+\vartheta\!\right\}_{\R^{2d}}\,\mathrm{d}\zeta,
\end{align*}
where we have used the fact that the Moyal bracket of two functions coincides with the canonical Poisson bracket whenever either of the two functions is a second-degree polynomial. In addition, we have
\begin{align*}
\int_{\R^{2d}} \mathcal{W}\left\{\frac{\delta F}{\delta \mathcal{W}},\, \frac{1}2\tr((\zeta\otimes\zeta)\,\Bbb{J}\mathcal{S}) +\xi\cdot\Bbb{J}^{T}\zeta+\vartheta\!\right\}_{\R^{2d}}\,\mathrm{d}\zeta
=&
\int_{\R^{2d}} \mathcal{W}\left(\nabla\frac{\delta F}{\delta\mathcal{W}}\right)\cdot\Bbb{J}\left(\Bbb{J}\mathcal{S}\zeta + \Bbb{J}\xi\right)\,\mathrm{d}\zeta
\\
=&
-\int_{\R^{2d}} \mathcal{W}\left(\nabla\frac{\delta F}{\delta\mathcal{W}}\right)\cdot\left(\mathcal{S}\zeta + \xi\right)\,\mathrm{d}\zeta
\\
=&
\int_{\R^{2d}} \frac{\delta F}{\delta\mathcal{W}}\left(\mathcal{S}\zeta +\xi\right)\cdot\nabla\mathcal{W}\,\mathrm{d}\zeta,
\end{align*}
where the last equality follows by integration by parts and by recalling that $\mathcal{S}\zeta$ is a Hamiltonian (divergenceless) vector field.

\section{Derivation of the Poisson Structure for Gaussian Moments}
\label{sec:PB-moments}
In this appendix, we derive the Poisson bracket~\eqref{eq:PB-GaussianWigner} on $\R^{2d} \times \sym(2d,\R)$ for the Gaussian moments from the Lie--Poisson bracket~\eqref{eq:PB-jac} on $\mathfrak{jac}(\R^{2d})^{*}$.
The first step is to employ the ``untangling map'' due to Krishnaprasad and Marsden \cite[Proposition 2.2]{KrMa1987}:
\begin{equation*}
  u\colon \widetilde{\mathfrak{jac}}(\R^{2d})^{*} \to \sp(\R^{2d}) \times \widetilde{\mathfrak{h}}(\R^{2d})^{*};
  \quad
  (\Pi, \lambda, \alpha) \mapsto \parentheses{ \Pi - \frac{\alpha^{-1}}{2} (\lambda \otimes \lambda)\mathbb{J}^{T}, \lambda, \alpha },
\end{equation*}
where we defined the open subsets
\begin{equation*}
  \widetilde{\mathfrak{jac}}(\R^{2d})^{*}
  \defeq \setdef{(\Pi, \lambda, \alpha) \in \mathfrak{jac}(\R^{2d})^{*}}{\alpha \neq 0},
  \qquad
  \widetilde{\mathfrak{h}}(\R^{2d})^{*}
  \defeq \setdef{(\lambda, \alpha) \in \mathfrak{h}(\R^{2d})^{*}}{\alpha \neq 0}
\end{equation*}
to avoid the singularity at $\alpha = 0$.
The untangling map $u$ is Poisson with respect to \eqref{eq:PB-jac} and the Poisson bracket
\begin{equation*}
  \PB{f}{g}_{\sp(\R^{2d})^{*} \times \widetilde{\mathfrak{h}}(\R^{2d})^{*}}(\tilde{\mu}, \lambda, \alpha) \defeq 
  \alpha \PB{f}{g}_{\R^{2d}}
  - \tr\parentheses{ \tilde{\mu}^{T} \brackets{ \Fd{f}{\tilde{\mu}}, \Fd{g}{\tilde{\mu}} } }
\end{equation*}
on $\sp(\R^{2d}) \times \widetilde{\mathfrak{h}}(\R^{2d})^{*} = \{ (\tilde{\mu}, \lambda, \alpha) \}$.
We then have
\begin{equation*}
  u \circ \mathbf{J}_{\mathfrak{s}^{*}}(\mathcal{W})
  = \parentheses{
    \frac{1}{2}\mathbb{J}^{T}\big( \exval{\zeta \otimes \zeta} - \exval{\zeta} \otimes \exval{\zeta} \big),\,
    \mathbb{J}^{T} \exval{\zeta},\,
    \exval{1}
  }.
\end{equation*}
Furthermore, we may identify $\sp(\R^{2d})^{*} \times \mathfrak{h}(\R^{2d})^{*}$ with $\R^{2d+1} \times \sym(2d,\R) = \{(\alpha, z, \mu)\}$ via the isomorphism
\begin{align*}
  \iota\colon \sp(\R^{2d})^{*} \times \mathfrak{h}(\R^{2d})^{*} &\to \R^{2d+1} \times \sym(2d,\R) \\
  (\tilde{\mu}, \lambda, \alpha) &\mapsto (\alpha, \mathbb{J}\lambda, \mathbb{J}\,\tilde{\mu}).
\end{align*}
See \citet{GaTr2012} for details on this isomorphism; the identification $\mathfrak{sp}(\Bbb{R}^{2d})^*\cong \sym(2d,\R)$ is explained in Section~\ref{ssec:sp-sym}.
As a result, we have
\begin{equation*}
  \iota \circ u \circ \mathbf{J}_{\mathfrak{s}^{*}}(\mathcal{W})
  = \parentheses{
    \exval{1},\,
    \exval{\zeta},\,
    \frac{1}{2}\parentheses{ \exval{\zeta \otimes \zeta} - \exval{\zeta} \otimes \exval{\zeta} }
  },
\end{equation*}
yielding the zeroth moment (as mentioned below, $\langle 1\rangle=1$ if $\mathcal{W}$ is normalized) as well as the first and second moments of $\mathcal{W}$ from \eqref{eq:moments}.
Let us write $\widetilde{\R}^{2d+1} \defeq \setdef{ (\alpha, z) }{ \alpha \neq 0 }$.
Then, restricting the map $\iota$ to $\sp(\R^{2d})^{*} \times \widetilde{\mathfrak{h}}(\R^{2d})^{*}$ and $\widetilde{\R}^{2d+1} \defeq \setdef{ (\alpha, z) }{ \alpha \neq 0 }$, we again obtain a Poisson map from $\sp(\R^{2d})^{*} \times \widetilde{\mathfrak{h}}(\R^{2d})^{*}$ to $\widetilde{\R}^{2d+1} \times \sym(2d,\R)$ with the Poisson bracket
\begin{equation}
  \label{eq:PB-GaussianWigner0}
  \PB{f}{g}_{\widetilde{\R}^{2d+1} \times \sym(2d,\R)}(\alpha, z, \mu) \defeq
  \alpha\PB{f}{g}_{\R^{2d}}
  - \tr\parentheses{
    \mu \brackets{
      \frac{\delta f}{\delta \mu}, \frac{\delta g}{\delta \mu}
    }_{\sym}
  },
\end{equation}
where $[\,\cdot\,, \,\cdot\,]_{\sym}$ is the Lie bracket on $\sym(2d,\R)$ defined in \eqref{eq:Lie_bracket-Sym} in Section~\ref{ssec:sp-sym}.

For example, for the Gaussian Wigner function $\mathcal{W}_{0}$ in \eqref{eq:W_0} with a positive-definite $2d \times 2d$ matrix $\Sigma$, we obtain
\begin{equation*}
  \iota \circ u \circ \mathbf{J}_{\mathfrak{s}^{*}}(\mathcal{W}_{0}) = \parentheses{ 1, z,\, \frac{\hbar}{4}\Sigma }.
\end{equation*}
This motivates us to reparametrize elements in $\widetilde{\R}^{2d+1} \times \sym(2d,\R)$ as $\{ (\alpha, z, \Sigma) \}$ with $\mu = \frac{\hbar}{4}\Sigma$.
Then the Poisson bracket~\eqref{eq:PB-GaussianWigner0} becomes
\begin{equation}
  \label{eq:PB-GaussianWigner1}
  \PB{f}{g}(\alpha,z,\Sigma) =
  \alpha\PB{f}{g}_{\R^{2d}}
  - \frac{4}{\hbar}\tr\parentheses{ \Sigma\! \left[ \Fd{f}{\Sigma}, \Fd{g}{\Sigma} \right]_{\sym} }.
\end{equation}
Now let us write $\mathbf{m} = (\alpha,z,\Sigma)$ for short.
Then, given a Hamiltonian $h\colon \widetilde{\R}^{2d+1} \times \sym(2d,\R) \to \R$, the Hamiltonian system $\dot{\mathbf{m}} = \PB{\mathbf{m}}{h}$ yields
\begin{equation*}
  \dot{\alpha} = 0,
  \qquad
  \dot{z} = \alpha\PB{z}{h}_{\R^{2d}}
  \qquad
  \dot{\Sigma} = \frac{4}{\hbar} \left( \mathbb{J}\Fd{h}{\Sigma}\Sigma - \Sigma\Fd{h}{\Sigma}\mathbb{J} \right).
\end{equation*}
If the Wigner function is normalized, one may set $\alpha = 1$ and so one may restrict the Poisson bracket~\eqref{eq:PB-GaussianWigner1} to $\R^{2d} \times \sym(2d,\R) = \{(z, \Sigma)\}$ to obtain the desired Poisson bracket~\eqref{eq:PB-GaussianWigner}.

\bibliography{GWP-Wigner}
\bibliographystyle{plainnat}

\end{document}